\documentclass[acmsmall,UKenglish]{acmart}

\AtBeginDocument{%
  \providecommand\BibTeX{{%
    \normalfont B\kern-0.5em{\scshape i\kern-0.25em b}\kern-0.8em\TeX}}}

\copyrightyear{2019}
\acmYear{2019}
\setcopyright{acmlicensed}
\acmConference[]{}{}{}
\acmBooktitle{}
\acmPrice{}
\acmDOI{}
\acmISBN{}

\usepackage{comment}
\usepackage{graphicx}

\usepackage{thmtools}
\usepackage{thm-restate}

\usepackage{tikz}
\usetikzlibrary{fit}
\usetikzlibrary{intersections}
\tikzstyle{vertex}=[circle,fill,inner sep=1pt]
\tikzstyle{point}=[circle,fill,inner sep=1pt]
\tikzstyle{spoint}=[circle,fill,inner sep=0.88pt]
\usetikzlibrary{calc}

\usepackage{xspace}
\usepackage{rotating}
\usepackage[color=green!20]{todonotes}

\def\msi{\textsc{Multicolored Subgraph Isomorphism}\xspace}

\def\mkc{\textsc{Multicolored $k$-Clique}\xspace}

\def\ksc{\textsc{$k$-Set Cover}\xspace}

\def\shs{\textsc{Structured 2-Track Hitting Set}\xspace}

\def\vgag{\textsc{Vertex Guard Art Gallery}\xspace}
\def\pgag{\textsc{Point Guard Art Gallery}\xspace}

\def\wone{$W[1]$\xspace}

\def\sti{two-block\xspace}

\newcommand{\rot}[1]{\begin{turn}{90}#1\end{turn}}

\def\asso{associate\xspace}

\def\ray{$\text{ray}$}

\def\seg{$\text{seg}$}
\def\segm{\nobreak\hspace{.25em}$\text{seg}$}

\title{Parameterized Hardness of Art Gallery Problems}

\author{\'Edouard Bonnet}
\authornote{supported by the LABEX MILYON (ANR-10- LABX-0070) of Universit\'e de Lyon, within the program ``Investissements d'Avenir'' (ANR-11-IDEX-0007) operated by the French National Research Agency (ANR).}
\email{edouard.bonnet@ens-lyon.fr}
\affiliation{%
\institution{Univ Lyon, CNRS, ENS de Lyon, Universit\'e Claude Bernard Lyon 1, LIP UMR5668}
}

\author{Tillmann Miltzow}
\authornote{supported by the ERC Consolidator Grant 615640-ForEFront. The author acknowledges generous support  from  the Netherlands Organisation for Scientific Research (NWO) under project no. 016.Veni.192.250.}
\email{t.miltzow@gmail.com}
\affiliation{%
  \institution{Utrecht University}
}


\setcopyright{acmcopyright}
\acmJournal{TALG}
\acmYear{2020} \acmVolume{1} \acmNumber{1} \acmArticle{1} \acmMonth{1} \acmPrice{15.00}\acmDOI{10.1145/3398684}


\begin{abstract}
  Given a simple polygon $\mathcal{P}$ on $n$ vertices, two points $x,y$ in $\mathcal{P}$ are said to be visible to each other if the line segment between $x$ and $y$ is contained in $\mathcal{P}$. 
  The 
  \pgag 
  problem 
  asks for a minimum set $S$ such that every point in $\mathcal{P}$ is visible from a point in $S$. 
  The 
  \vgag
  problem asks for such a set $S$ subset of the vertices of $\mathcal{P}$. 
  A point in the set $S$ is referred to as a guard. 
  For both variants, we rule out any $f(k)n^{o(k / \log k)}$ algorithm, where $k := |S|$ is the number of guards, for any computable function $f$, unless the Exponential Time Hypothesis fails. 
  These lower bounds almost match the $n^{O(k)}$ algorithms that exist for both problems.
\end{abstract}

\begin{document}

\maketitle

\section{Introduction}
  Two points $x,y$ in a simple polygon $\mathcal{P}$ are said to be visible to each other if the line segment between $x$ and $y$ is contained in $\mathcal{P}$. 
  The
  \pgag 
  problem asks for a minimum set $S$ such that every point in $\mathcal{P}$ is visible from a point in $S$. 
  The 
  \vgag 
  problem asks for such a set $S$ subset of the vertices of $\mathcal{P}$.
  In both cases, such a set $S$ is a \emph{guarding set} and its elements are called \emph{guards}.
  In the decision versions, given a simple polygon and an integer, one has to decide if there is a guarding set for the polygon of cardinality at most the integer.
  In what follows, $n$ refers to the number of vertices of $\mathcal{P}$ and $k$ to the allowed number of guards.
  
The art gallery problem is arguably one of the most well-known problems in discrete and computational geometry. 
Since its introduction by Viktor Klee in 1976, numerous research papers were published on the subject.
O'Rourke's early book from 1987 \cite{o1987art} has over two thousand citations, and each year, top conferences publish new results on the topic. 
Many variants of the art gallery problem, based on different definitions of visibility, restricted classes of polygons, different shapes of guards, have been defined and analyzed. 
  One of the first results is the elegant proof of Fisk that $\lfloor n/3 \rfloor$ guards are always sufficient and sometimes necessary for a polygon with $n$ vertices~\cite{DBLP:journals/jct/Fisk78a}.

  The art gallery problem was shown NP-hard by Aggarwal in his PhD thesis \cite{aggarwal84} and by Lee and Lin \cite{lee86}.
  Eidenbenz et al.~\cite{eidenbenz2001inapproximability} even showed APX-hardness for the most standard variants. 
  See also~\cite{katz2008guarding,broden2001guarding, DBLP:journals/algorithmica/KrohnN13} for other hardness constructions.
  Very recently, Abrahamsen et al. \cite{Abrahamsen18} showed that \pgag is $\exists \mathbb R$-complete.
  In particular, this problem is unlikely to be in NP. This is maybe intuitive, if we consider simple instances of the art gallery 
  problem, which need irrational numbers for an optimal guard placement~\cite{IrrationalGuards}.
  In contrast, Dobbins, Holmsen and Miltzow~\cite{Dobbins18} showed how to find a solution with rational coordinates using the concept of smoothed analysis.
Due to those negative results, most papers focus on finding approximation algorithms and on variants or restrictions that are polynomially tractable~\cite{ghosh2010approximation, DBLP:journals/dcg/Kirkpatrick15, DBLP:journals/comgeo/King13, DBLP:journals/jcss/MotwaniRS90, DBLP:journals/algorithmica/KrohnN13}.
For the \pgag problem on simple polygons, there is an $O(\log{\text{OPT}})$-approximation under some assumptions (integer coordinates and some special \emph{general position} of the vertices) \cite{BonnetM17}.
The approximation relies on the construction of $\varepsilon$-nets and ideas from Efrat and Har-Peled \cite{DBLP:journals/ipl/EfratH06}.
For polygons with $h$ holes, there is a polynomial approximation algorithm with ratio $O(\log{\text{OPT}} \cdot \log {h})$ which guards all but a $\delta$-fraction of the polygon \cite{Elbassioni17}.
Recently, a constant-factor approximation was announced for \vgag \cite{Bhattacharya17}.
However, a mistake was later found~\cite{AshurWeakVertexApprox}. 
   Another approach is to find heuristics to solve large instances of the art gallery problem~\cite{de2016engineering}. 
   Naturally, the fundamental drawback of this approach is the lack of performance guarantees.

In the last twenty-five years, another fruitful approach gained popularity: parameterized complexity. 
The underlying idea is to study algorithmic problems with dependence on a natural parameter. 
If the dependence on the parameter is practical and the parameter is small for real-life instances, we attain algorithms that give optimal solutions with reasonable running times. 
  For a gentle introduction to parameterized complexity, we recommend Niedermeier's book~\cite{Niedermeier2006}.
  For a thorough reading highlighting complexity classes, we suggest the book by Downey and Fellows~\cite{downey2012parameterized}.
  For a recent book on the topic with an emphasis on algorithms, we advise to read the book by Cygan et al.~\cite{DBLP:books/sp/CyganFKLMPPS15}.
  An approach based on logic is given by Flum and Grohe~\cite{flum2006parameterized}.
  Despite the recent successes of parameterized complexity, only very few results on 
  the art gallery problem are known prior to this paper.
%

  The first such result is the trivial algorithm 
  for the vertex guard variant to check if a solution of size $k$ exists in a polygon with $n$ vertices. 
  The algorithm runs in $O(n^{k+2})$ time, by checking 
  all possible subsets of size $k$ of the vertices. 
  The second \emph{not so well-known} result is the fact that one can find in time $n^{O(k)}$ a set of $k$ guards for the point guard variant, if it exists~\cite{DBLP:journals/ipl/EfratH06}, using tools from real algebraic geometry~\cite{basu2011algorithms}.  
  This was first observed by Sharir~\cite[Acknowledgment]{DBLP:journals/ipl/EfratH06}.
  Despite the fact that the first algorithm is extremely basic and the second algorithm, even with remarkably sophisticated tools, uses almost no problem specific insights, no better exact parameterized algorithms are known.

The Exponential Time Hypothesis (ETH) asserts that there is no $2^{o(N)}$ time algorithm for \textsc{Sat} on $N$ variables. 
The ETH is used to attain more precise conditional lower bounds than the mere NP-hardness.
  A simple reduction from \textsc{Set Cover} by Eidenbenz et al. shows that there is no $f(k)n^{o(k)}$ algorithm for these problems, when we consider polygons with holes~\cite[Sec.4]{eidenbenz2001inapproximability}, unless the ETH fails.
  However, polygons with holes are very different from simple polygons. 
For instance, they have unbounded VC-dimension while simple polygons have bounded VC-dimension~\cite{valtr1998guarding, gilbers2014new,kalai1997guarding, gibson2015vc}.

  
We present the first lower bounds for the parameterized art gallery problems restricted to \emph{simple} polygons. 
Here, the parameter is the optimal number $k$ of guards to cover the polygon. 

  \begin{restatable}[Parameterized hardness point guard]{theorem}{ParaPoint}
    \label{thm:ParaPoint}
    \pgag 
    is not solvable in time $f(k)n^{o(k / \log k)}$, even on simple polygons, where $n$ is the number of vertices of the polygon and $k$ is the number of guards allowed, for any computable function $f$, unless the ETH fails. 
  \end{restatable}
   
  \begin{restatable}[Parameterized hardness vertex guard]{theorem}{ParaVertex}
    \label{thm:ParaVertex}
    \vgag 
     is not solvable in time $f(k)n^{o(k / \log k)}$, even on simple polygons, where $n$ is the number of vertices of the polygon and $k$ is the number of guards allowed, for any computable function $f$, unless the ETH fails. 
  \end{restatable}
  
  These results imply that the previous noted algorithms are essentially tight, and suggest that there are no significantly better parameterized algorithms. 
Our reductions are from \textsc{Subgraph Isomorphism} and therefore an $f(k)n^{o(k)}$-algorithm for the art gallery problem would also imply improved algorithms for \textsc{Subgraph Isomorphism} and for CSP parameterized by treewidth, which would be considered a major breakthrough~\cite{Marx10}.
Let us also mention that our results imply that both variants are \wone-hard parameterized by the number of guards.

After the conference version of this paper appeared, the parameterized complexity of the art gallery and related problems was investigated further.
The parameterized complexity of the terrain guarding problem was studied~\cite{ashok18}.
The terrain guarding problem is a particular case of the art gallery problem, where instead of a polygon, one should guard an $x$-monotone curve.
This restriction is still NP-hard~\cite{King11}, even on rectilinear (that is, every edge is horizontal or vertical) terrains~\cite{bonnet2019orthogonal}.
The authors of \cite{ashok18} present an $n^{O(\sqrt{k})}$-time algorithm (hence $2^{O(n^{1/2} \log n)}$) for guarding general $n$-vertex terrains with $k$ guards, and an FPT $k^{O(k)}n^{O(1)}$-time algorithm for guarding the vertices of rectilinear terrains.
Note that there is no $2^{o(n^{1/3})}$ algorithm for terrain guarding, unless the ETH fails~\cite{bonnet2019orthogonal}.

The art gallery problem parameterized by the number of reflex vertices is considered by Agrawal et al.~\cite{agrawal2020parameterized}.
The authors present an FPT algorithm for \vgag under this parameterization. 
See also~\cite{agrawal2019fpt} for FPT algorithms on the (strong) conflict-free coloring of terrains.


%

\section{Proof Ideas}\label{subsec:proof-idea}

In order to achieve these results, we slightly extend some known hardness results of geometric set cover/hitting set problems and combine them with problem-specific insights of the art gallery problem.
  One of the first problem-specific insights is the ability to encode \textsc{Hitting Set} on interval graphs. 
  The reader can refer to Figure~\ref{fig:IntervalEncoding} for the following description.
  Assume that we have some fixed points $p_1,\ldots,p_n$ with increasing $y$-coordinates in the plane. 
We can build a pocket ``far enough to the right'' that can be seen only from $\{p_i,\ldots,p_j\}$ for any $1\leqslant i<j\leqslant n$. 
  \begin{figure}[htpb]
\centering
\includegraphics{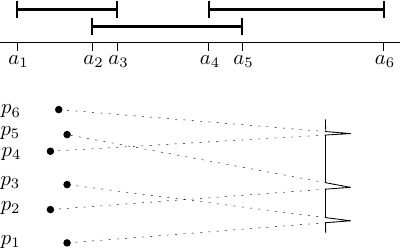}
\caption{Reduction from \textsc{Hitting Set} on interval graphs to a restricted version of the art gallery problem.}
\label{fig:IntervalEncoding}
\end{figure}

Let $I_1,\ldots,I_n$ be $n$ intervals with endpoints $a_1,\ldots, a_{2n}$.
Then, we construct $2n$ points $p_1,\ldots, p_{2n}$ representing $a_1,\ldots, a_{2n}$. 
Further, we construct one pocket ``far enough to the right'' for each interval as described above. 
This way, we reduce \textsc{Hitting Set} on interval graphs to a restricted version of the art gallery problem.
This observation is \emph{not} so useful in itself since \textsc{Hitting Set} on interval graphs can be solved in polynomial time. 

\begin{figure}[htpb]
  \centering
  \includegraphics{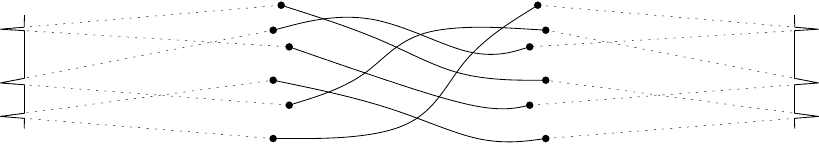}
  \caption{Two instances of Hitting Set ``magically'' linked.}
  \label{fig:MagicLinking}
\end{figure}

The situation changes rapidly if we consider \textsc{Hitting Set} on $2$-track interval graphs, as described in the preliminaries. 
Unfortunately, we are not able to just ``magically'' link (see Figure~\ref{fig:MagicLinking}) some specific pairs of points in the polygon of the art gallery instance. 
Instead, we construct linking gadgets, which work ``morally'' as follows. 
We are given two set of points $P$ and $Q$ and a bijection $\sigma$ between $P$ and $Q$. 
The linking gadget is built in a way that it can be covered by two points $(p,q)$ of $P \times Q$, if and only if $q = \sigma(p)$.
The \shs problem will be specifically designed so that the linking gadget is the main remaining ingredient to show hardness.
This intermediate problem is a convenient starting point for parameterized reductions to other geometric problems.
For instance, the parameterized hardness of \textsc{Red-Blue Points Separation}, where given a set of blue points and a set of red points in the plane, one has to find at most $k$ lines so that no cell of the arrangement is bichromatic, was obtained by a reduction from \shs \cite{BonnetGL17}.

\textbf{Organization.}
The rest of the paper is organized as follows.
In Section~\ref{sec:prelim}, we introduce some notations, discuss the encoding of the polygon, give some useful ETH-based lower bounds, and prove a technical lemma.
In Section~\ref{sec:intermediate}, we prove the lower bound for \shs (Theorem~\ref{cor:main}). 
Lemma~\ref{thm:2-intervals-cover} contains the key arguments.
From this point onward, we can reduce from \shs. 
In Section~\ref{sec:point-guard}, we show the lower bound for the \pgag problem (Theorem~\ref{thm:ParaPoint}).
We design a linking gadget, show its correctness, and show how several linking gadgets can be combined consistently. 
In Section~\ref{sec:vertex-guard}, we tackle the \vgag problem (Theorem~\ref{thm:ParaVertex}).
We have to design a very different linking gadget, that has to be combined with other gadgets and ideas.

\section{Preliminaries}\label{sec:prelim}

For any two integers $x \leqslant y$, we set $[x,y]:=\{x,x+1,\ldots,y-1,y\}$, and for any positive integer $x$, $[x]:=[1,x]$.
Given two points $a,b$ in the plane, we define $\seg(a,b)$ as the line segment with endpoints $a,b$. 
Given $n$ points $v_1,\ldots,v_n \in \mathbb{R}^2$, we define a polygonal closed curve $c$ by 
$\seg(v_1,v_2)$, $\ldots,$ $\seg(v_{n-1},v_n)$, $\seg(v_n,v_1)$. If $c$ is not self intersecting, it partitions the plane into a closed bounded area and an unbounded area. 
The closed bounded area is a \emph{simple polygon} on the vertices $v_1,\ldots,v_n$. 
Note that we do not consider the boundary as the polygon but rather all the points bounded by the curve $c$ as described above.
Given two points $a,b$ in a simple polygon $\mathcal{P}$, we say that $a$ \emph{sees} $b$ or $a$ is {\em visible} from $b$ if $\seg(a,b)$ is contained in $\mathcal{P}$.
By this definition, it is possible to ``see through'' vertices of the polygon.
We say that $S$ is a set of \emph{point guards} of $\mathcal{P}$, if every point $p\in \mathcal{P}$ is visible from a point of $S$.
We say that $S$ is a set of \emph{vertex guards} of $\mathcal{P}$, if additionally $S$ is a subset of the vertices of $\mathcal{P}$.
The \pgag problem  and the \vgag problem are formally defined as follows.
\begin{quote}
   \noindent \textbf{\pgag}
 \\  \textbf{Input:} The vertices of a simple polygon $\mathcal{P}$ in the plane and a natural number $k$. \\
\textbf{Question:} Does there exist a set of $k$ point guards for $\mathcal{P}$?
\end{quote}
\begin{quote}
   \noindent \textbf{\vgag}
 \\  \textbf{Input:} A simple polygon $\mathcal{P}$ on $n$ vertices in the plane and a natural number $k$. \\
\textbf{Question:} Does there exist a set of $k$ vertex guards for
$\mathcal{P}$?
\end{quote}

For any two distinct points $v$ and $w$ in the plane we denote by $\ray(v,w)$ the ray starting at $v$ and passing through $w$, and by $\ell(v,w)$ the supporting line passing through $v$ and $w$.  
For any point $x$ in a polygon $\mathcal P$, $V_{\mathcal P}(x)$, or simply $V(x)$, denotes the \emph{visibility region} of $x$ within $\mathcal P$, that is the set of all the points $y \in \mathcal P$ seen by $x$.
We say that two vertices $v$ and $w$ of a polygon $\mathcal P$ are \emph{neighbors} or \emph{consecutive} if $vw$ is an edge of $\mathcal P$.
A \emph{sub-polygon} $\mathcal P'$ of a simple polygon $\mathcal P$ is defined by any $l$ distinct consecutive vertices $v_1, v_2, \ldots, v_l$ of $\mathcal P$ (that is, for every $i \in [l-1]$, $v_i$ and $v_{i+1}$ are neighbors in $\mathcal P$) such that $v_1v_l$ does not cross any edge of $\mathcal P$. 
In particular, $\mathcal P'$ is a simple polygon.

\textbf{Encoding.}
We assume that the vertices of the polygon are either given by integers or by rational numbers.
We also assume that the output is given either by integers or by rational numbers. 
The instances we generate as a result of Theorem~\ref{thm:ParaPoint} and Theorem~\ref{thm:ParaVertex} have rational coordinates.
We can represent each coordinate by specifying the nominator and denominator.
The number of bits is bounded by $O(\log n)$ in both cases. We can transform the coordinates to integers by multiplying every coordinate with the least common multiple of all denominators. However, this leads to integers using $O(n\log n)$ bits.



\textbf{ETH-based lower bounds.}
The \emph{Exponential Time Hypothesis} (ETH) is a conjecture by Impagliazzo et al.~\cite{impagliazzo1999complexity} asserting that there is no $2^{o(n)}$-time algorithm for \textsc{3-SAT} on instances with $n$ variables.
The $k$-\textsc{Multicolored-Clique} problem has as input a graph $G=(V,E)$, where the set of vertices is partitioned into $V_1,\ldots,V_k$. It asks if there exists a set of $k$ vertices $v_1\in V_1,\ldots,v_k\in V_k$ such that these vertices form a clique of size $k$.
We will use the following lower bound proved by Chen et al.~\cite{chen2006strong}.
\begin{theorem}[\cite{chen2006strong}]\label{thm:chen06}
There is no $f(k)n^{o(k)}$ algorithm for $k$-\textsc{Multicolored-Clique}, for any computable function $f$, unless the ETH fails.
\end{theorem}

Marx showed that \textsc{Subgraph Isomorphism} cannot be solved in time $f(k)n^{o(k / \log k)}$ where $k$ is the number of edges of the pattern graph, under the ETH \cite{Marx10}.
Usually, this result enables to improve a lower bound obtained by a reduction from \mkc with a quadratic blow-up on the parameter, from exponent $o(\sqrt k)$ to exponent $o(k / \log k)$, by doing more or less the same reduction but from \textsc{Multicolored Subgraph Isomorphism}.
In the \textsc{Multicolored Subgraph Isomorphism} problem, one is given a graph with $n$ vertices partitioned into $l$ color classes $V_1, \ldots, V_l$ such that only $k$ of the ${l \choose 2}$ sets $E_{ij}=E(V_i,V_j)$ are non empty.
The goal is to pick one vertex in each color class so that the selected vertices induce $k$ edges.
The technique of color coding and the result of Marx shows that:
\begin{theorem}[\cite{Marx10}]\label{thm:marx10}
\textsc{Multicolored Subgraph Isomorphism} cannot be solved in time $f(k)n^{o(k / \log k)}$ where $k$ is the number of edges of the solution, for any computable function $f$, unless the ETH fails.
\end{theorem}
Naturally, this result still holds when restricted to connected input graphs.
In that case, $k \geqslant l-1$.


\textbf{Bounding the coordinates.}
We say a point $p = (p_x,p_y) \in \mathbb{Z}^2$ has coordinates bounded  by $L$ if $|p_x|,|p_y|\leqslant L$.
Given two vectors $v,w$, we denote their scalar product as $v\cdot w$.
This technical lemma will prove useful to ensure that the polygon built in Section~\ref{sec:point-guard} can be described with integer coordinates.

\begin{lemma}\label{lem:Rational}
  Let $p^1,q^1,p^2,q^2$ be four points  
  with integer coordinates bounded by $L$.
  Then the intersection point $d = (d_x,d_y)$ of the  supporting lines $\ell_1 = \ell(p^1,q^1)$ and $\ell_2  = \ell(p^2,q^2)$ is a rational point. The nominator and denominator of $d_x$ and $d_y$ are bounded by $O(L^2)$.
\end{lemma}
\begin{proof}
  The fact that $d$ lies on $\ell_i$ can be expressed as
  $v_i \cdot d = b_i$,
  with some appropriate vector $v^i$ and number $b^i$, for $i=1,2$.
  To be precise $v^i = (- p^i_x+q^i_x,  p^i_y-q^i_y)$ and $b^i = v_i \cdot p^i$, for $i=1,2$.
  We define the matrix $A = (v^1,v^2)$ and the vector $b = (b^1,b^2)$.
  Then both conditions can be expressed as
  $A \cdot d = b$.
  We denote by $A_i$ the matrix $i$ with the $i$-th column replaced by $b$.
  And by $\textup{det}(M)$ the determinant of the matrix $M$.
  By Cramer's rule, it holds that 
  $d_x = \frac{\textup{det}(A_1)}{\textup{det}(A)}$ and
  $d_y = \frac{\textup{det}(A_2)}{\textup{det}(A)}$.
\end{proof}




\section{Parameterized hardness of \shs}\label{sec:intermediate}

The purpose of this section is to show Theorem~\ref{cor:main}.
As we will see at the end of the section, there already exist quite a few parameterized hardness results for set cover/hitting set problems restricted to instances with some geometric flavor.
The crux of the proof of Theorem~\ref{cor:main} lies in Lemma~\ref{thm:2-intervals-cover}.
We introduce a few notation and vocabulary to state and prove this lemma.

Given a finite totally ordered set $Y=\{y_1, \ldots, y_{|Y|}\}$ (that is, for any $i, j \in [|Y|]$, $y_i \leq y_j$ iff $i \leqslant j$), a subset $S \subseteq Y$ is a \emph{$Y$-interval} if $S=\{y $ $|$ $ y_i \leq y \leq y_j\}$ for some $i$ and $j$.
We denote by $\leq_Y$ the order of $Y$.
A set-system $(X,\mathcal S)$ is said to be \emph{\sti} if $X$ can be partitioned into two totally ordered sets $A=\{a_1, \ldots, a_{|A|}\}$ and $B=\{b_1, \ldots, b_{|B|}\}$ such that each set $S \in \mathcal S$ is the union of an $A$-interval with a $B$-interval.

Given a set $\mathcal S$ of subsets of $X$, \ksc asks to find $k$ sets of $\mathcal S$ whose union is $X$.
We first show an ETH lower bound and W[1]-hardness for \ksc restricted to \sti instances.
We reduce from \mkc for simplicity sake (then from \msi to improve the ETH lower bound).
On a high-level, we encode adjacencies in the \mkc instance by pairs of disjoint sets particularly effective to cover $X$.
On the contrary, pairs of non-adjacent vertices will be mapped to pairs of sets overlapping and missing an important part of $X$.
This trick will be a recurring theme throughout the paper.

\begin{lemma}\label{thm:2-intervals-cover}
\ksc restricted to \sti instances with $N$ elements and $M$ sets is \wone-hard and not solvable in time $f(k)(N+M)^{o(k / \log k)}$ for any computable function $f$, unless the ETH fails.
\end{lemma}

\begin{proof}
We reduce from \mkc which remains \wone-hard when each color class has the same number $t$ of vertices. 
Let $G = (V_1 \cup \ldots \cup V_k, E)$ be an instance of \mkc with $V = \bigcup_{i \in [k]} V_i$, $\forall i \in [k]$, $V_i=\{v^i_1, \ldots, v^i_t\}$, $m=|E|$, and $n=|V|=tk$.
For each pair $i < j \in [k]$\footnote{By $i < j \in [k]$, we mean that $i \in [k], j \in [k]$, and $i<j$.}, $E_{ij}$ denotes the set of edges $E(V_i,V_j)$ between $V_i$ and $V_j$.
For each $E_{ij}$ we give an arbitrary order to the edges: $e^{ij}_1, \ldots, e^{ij}_{|E_{ij}|}$.
We build an equivalent instance $(X, \mathcal S)$ of \ksc with $4{k \choose 2}+4m+tk(k+1)+4k$ elements and $4m+2kt$ sets, and such that $(X, \mathcal S)$ is \sti.
We call $A$ and $B$ the two sets of the partition of $X$ that realizes that $(X, \mathcal S)$ is \sti.

For each of the color class $V_i$, we add $tk+2$ elements to $A$ with the following order:
\[x_b(i),\]
\[x(i,1,1), \ldots, x(i,1,t),\]
\[x(i,2,1), \ldots, x(i,2,t),\]
\[\ldots\]
\[x(i,i-1,1), \ldots, x(i,i-1,t),\]
\[x(i,i+1,1), \ldots, x(i,i+1,t),\]
\[\ldots\]
\[x(i,k+1,1), \ldots, x(i,k+1,t),\]
\[x_e(i),\] and call $X(i)$ the set containing those elements.
We also set $$X(i,j) := \{x(i,j,1), x(i,j,2), \ldots, x(i,j,t)\}$$ (hence, $X(i)=\bigcup_{j \neq i}X(i,j) \cup \{x_b(i),x_e(i)\}$).
For each $E_{ij}$, we add to $B$ the $3|E_{ij}|+2$ of a set $Y(i,j)$ ordered: $$y_b(i,j), y(i,j,1), \ldots, y(i,j,3|E_{ij}|), y_e(i,j).$$
For each pair $i < j \in [k]$ and for each edge $e^{ij}_c=v^i_av^j_b$ in $E_{ij}$ (with $a, b \in [t]$ and $c \in [|E_{ij}|]$), we add to $\mathcal S$ the two sets $$S(e^{ij}_c,v^i_a):=\{x(i,j,a), x(i,j,a+1), \ldots, x(i,j,t), x(i,j+1,1), \ldots, x(i,j+1,a-1)\}$$ $$\cup~\{y(i,j,c), \ldots, y(i,j,c+|E_{ij}|-1)\}~\text{and}$$  $$S(e^{ij}_c,v^j_b):=\{x(j,i,b), x(j,i,b+1), \ldots, x(j,i,t), x(j,i+1,1), \ldots x(j,i+1,b-1)\}$$ $$\cup~\{y(i,j,c+|E_{ij}|), \ldots, y(i,j,c+2|E_{ij}|-1)\}.$$
Observe that in case $j=i+1$, then all the elements of the form $x(j,i+1,\cdot)$ in set $S(e^{ij}_c,v^j_b)$ are in fact of the form $x(j,i+2,\cdot)$.
We may also notice that in case $a=1$ (resp. $b=1$), then there is no element of the form $x(i,j+1,\cdot)$ (resp. $x(j,i+1,\cdot)$) in set $S(e^{ij}_c,v^i_a)$ (resp. in set $S(e^{ij}_c,v^j_b)$).
For each pair $i < j \in [k]$, we also add to $A$ the $|E_{ij}|+2$ elements of a set $Z(i,j)$ ordered: $$z_b(i,j), z(i,j,1), \ldots, z(i,j,|E_{ij}|), z_e(i,j),$$ and for each edge $e^{ij}_c$ in $E_{ij}$ (with $c \in [|E_{ij}|]$), we add to $\mathcal S$ the two sets $$S(e^{ij}_c,\vdash)=\{z_b(i,j), z(i,j,1), \ldots, z(i,j,|E_{ij}|-c\} \cup \{y_b(i,j), y(i,j,1) \ldots y(i,j,c-1)\}~\text{and}$$  $$S(e^{ij}_c,\dashv)=\{z(i,j,|E_{ij}|-c+1), \ldots, z(i,j,|E_{ij}|, z_e(i,j)\} \cup \{y(i,j,c+2|E_{ij}|) \ldots y(i,j,3|E_{ij}|), y_e(i,j)\}.$$
Finally, for each $i \in [k]$, we add to $B$ the $t+2$ elements of a set $W(i)$ ordered: $$w_b(i), w(i,1), \ldots, w(i,t), w_e(i),$$ and for all $a \in [t]$, we add the sets $$S(i,a,\vdash):=\{x_b(i), x(i,1,1), \ldots, x(i,1,a-1)\} \cup \{w_b(i), w(i,1), \ldots, w(i,t-a+1)\}~\text{and}$$ $$S(i,a,\dashv) :=\{x(i,k+1,a), \ldots, x(i,k+1,t), x_e(i)\} \cup \{w(i,t-a+2), \ldots, w(i,t), w_e(i)\}.$$

No matter the order in which we put the $X(i)$'s and $Z(i,j)$'s in $A$ (respectively the $Y(i,j)$'s and $W(i)$'s in $B$), the sets we defined are all unions of an $A$-interval with a $B$-interval, provided we keep the elements within each $X(i)$, $Z(i,j)$, $Y(i,j)$, and $W(i)$ consecutive (and naturally, in the order we specified).
Though, to clarify the construction, we fix the following orders for $A$ and for $B$:
$$X(1), \ldots, X(k), Z(1,2), \ldots, Z(1,k), Z(2,3), \ldots, Z(2,k), \ldots, Z(k-2,k-1), Z(k-2,k), Z(k-1,k)$$
$$Y(1,2), \ldots, Y(1,k), Y(2,3), \ldots, Y(2,k), \ldots, Y(k-2,k-1), Y(k-2,k), Y(k-1,k), W(1), \ldots, W(k).$$
We ask for a set cover with $2k^2$ sets.
This ends the construction (see Figure~\ref{fig:set-cover-2-block} for an illustration of the construction for the instance graph of Figure~\ref{fig:instance-of-mkc}).

\begin{figure}
\centering
\begin{tikzpicture}
\node (tv1) at (1,1.6) {$G$} ;

\node[draw,circle,inner sep=0.01cm] (v11) at (0,0) {$v^1_1$} ;
\node[very thick,draw,circle,inner sep=0.01cm] (v12) at (0,1) {$v^1_2$} ;
\node[draw, rectangle, rounded corners, fit=(v11) (v12)] (V1) {} ;
\node (tv1) at (0,-0.6) {$V_1$} ;

\node[draw,circle,inner sep=0.01cm] (v21) at (1,0) {$v^2_1$} ;
\node[very thick,draw,circle,inner sep=0.01cm] (v22) at (1,1) {$v^2_2$} ;
\node[draw, rectangle, rounded corners, fit=(v21) (v22)] (V2) {} ;
\node (tv2) at (1,-0.6) {$V_2$} ;

\node[draw,circle,inner sep=0.01cm] (v31) at (2,0) {$v^3_1$} ;
\node[draw,circle,inner sep=0.01cm] (v32) at (2,1) {$v^3_2$} ;
\node[draw, rectangle, rounded corners, fit=(v31) (v32)] (V3) {} ;
\node (tv3) at (2,-0.6) {$V_3$} ;

\draw (v11) -- (v22) ;
\draw (v11) -- (v32) ;
\draw (v21) -- (v31) ;
\draw[very thick] (v12) -- (v22) ;
\draw (v12) -- (v31) ;
\draw (v22) -- (v31) ;
\draw[very thick] (v12) -- ($(v31)!1cm!(v12)$) ;
\draw[very thick] (v22) -- ($(v31)!0.7cm!(v22)$) ;
\end{tikzpicture}
\caption{A simple instance of \mkc. The elements in bold: vertices $v^1_2$ and $v^2_2$, edge $v^1_2v^2_2$, and half of the edges $v^1_2v^3_1$ and $v^2_2v^3_1$ correspond to the selection of sets depicted in Figure~\ref{fig:set-cover-2-block}.}
\label{fig:instance-of-mkc}
\end{figure}

\begin{figure}
\setlength{\tabcolsep}{0.2pt}
\resizebox{\columnwidth}{!}
{
\begin{tabular}{ccccccccc|cccccccc|c|cccc|c||cccccccc|c|cccc|cccc|cc}

& \rot{$x_b(1)$} & \rot{$x(1,2,1)$} & \rot{$x(1,2,2)$} & \rot{$x(1,3,1)$} & \rot{$x(1,3,2)$} & \rot{$x(1,4,1)$} & \rot{$x(1,4,2)$} & \rot{$x_e(1)$} & \rot{$x_b(2)$} & \rot{$x(2,1,1)$} & \rot{$x(2,1,2)$} & \rot{$x(2,3,1)$} & \rot{$x(2,3,2)$} & \rot{$x(2,4,1)$} & \rot{$x(2,4,2)$} & \rot{$x_e(2)$} & $\ldots$ & \rot{$z_b(1,2)$} & \rot{$z(1,2,1)$} & \rot{$z(1,2,2)$} & \rot{$z_e(1,2)$} & $\ldots$ & \rot{$y_b(1,2)$} & \rot{$y(1,2,1)$} & \rot{$y(1,2,2)$} & \rot{$y(1,2,3)$} & \rot{$y(1,2,4)$} & \rot{$y(1,2,5)$} & \rot{$y(1,2,6)$} & \rot{$y_e(1,2)$} & $\ldots$ & \rot{$w_b(1)$} & \rot{$w(1,1)$} & \rot{$w(1,2)$} & \rot{$w_e(1)$} & \rot{$w_b(2)$} & \rot{$w(2,1)$} & \rot{$w(2,2)$} & \rot{$w_e(2)$} & $\ldots$ \\
$S(1,1,\vdash)$ & 1 & & & & & & & & & & & & & & & & & & & & & & & & & & & & & & & 1 & 1 & 1 & & & & & & & \\
$\mathbf{S(1,2,\vdash)}$ & \textbf{1} & \textbf{1} & & & & & & & & & & & & & & & & & & & & & & & & & & & & & & \textbf{1} & \textbf{1} & & & & & & & & \\
$S(v^1_1v^2_2,v^1_1)$ & & 1 & 1 & & & & & & & & & & & & & & & & & & & & & 1 & 1 & & & & & & & & & & & & & & &  \\
$\mathbf{S(v^1_2v^2_2,v^1_2)}$ & & & \textbf{1} & \textbf{1} & & & & & & & & & & & & & & & & & & & & &\textbf{1} & \textbf{1} & & & & & & & & & & & & & & \\
$S(v^1_1v^3_2,v^1_1)$ & & & & 1 & 1 & & & & & & & & & & & & & & & & & & & & & & & & & & & & & & & & & & & & \\
$\mathbf{S(v^1_2v^3_1,v^1_2)}$ & & & & & \textbf{1} & \textbf{1} & & & & & & & & & & & & & & & & & & & & & & & & & & & & & & & & &\\
$S(1,1,\dashv)$ & & & & & & 1 & 1 & 1 & & & & & & & & & & & & & & & & & & & & & & & & & & & 1 & & & &\\
$\mathbf{S(1,2,\dashv)}$ & & & & & & & \textbf{1} & \textbf{1} & & & & & & & & & & & & & & & & & & & & & & & & & & \textbf{1} & \textbf{1} & & & &\\
\hline
$S(2,1,\vdash)$ & & & & & & & & & 1 & & & & & & & & & & & & & & & & & & & & & & & & & & & 1 & 1 & 1 & & \\
$\mathbf{S(2,2,\vdash)}$ & & & & & & & & & \textbf{1} & \textbf{1} & & & & & & & & & & & & & & & & & & & & & & & & & & \textbf{1} & \textbf{1} & &  & & \\
$S(v^2_2v^1_1,v^2_2)$ & & & & & & & & & & & 1 & 1 & & & & & & & & & & & & & & 1 & 1 & & & & & & & & &  & & & &\\
$\mathbf{S(v^2_2v^1_2,v^2_2)}$ & & & & & & & & & & & \textbf{1} & \textbf{1} & & & & & & & & & & & & & & & \textbf{1} & \textbf{1} & & & & & & & & & & & \\
$S(v^2_1v^3_1,v^2_1)$ & & & & & & & & & & & & 1 & 1 & & & & & & & & & & & & & & & & & & & & & &  & & & &  \\
$\mathbf{S(v^2_2v^3_1,v^2_2)}$ & & & & & & & & & & & & & \textbf{1} & \textbf{1} & & & & & & & & & & & & & & & & & & & & & & & & &  \\
$S(2,1,\dashv)$ & & & & & & & & & & & & & & 1 & 1 & 1 & & & & & & & & & & & & & & & & & & & & & & & 1 \\
$\mathbf{S(2,2,\dashv)}$ & & & & & & & & & & & & & & & \textbf{1} & \textbf{1} & & & & & & & & & & & & & & & & & & & & & & \textbf{1} & \textbf{1} \\
\hline
$\mathbf{S(v^1_2v^2_2,\vdash)}$ & & & & & & & & & & & & & & & & & & \textbf{1} & & & & & \textbf{1} & \textbf{1} & & & & & & & & & & & & & & & & \\
$S(v^1_1v^2_2,\vdash)$ & & & & & & & & & & & & & & & & & & 1 & 1 & & & & 1 & & & & & & & & & & & & & & & & & & \\
$\mathbf{S(v^1_2v^2_2,\dashv)}$ & & & & & & & & & & & & & & & & & & & \textbf{1} & \textbf{1} & \textbf{1} & & & & & & & & \textbf{1} & \textbf{1} & & & & & &  & & & &\\
$S(v^1_1v^2_2,\dashv)$ & & & & & & & & & & & & & & & & & & & & 1 & 1 & & & & & & & 1 & 1 & 1 & & & & & & & & & &\\

\end{tabular}
}
\caption{The sets of $\mathcal S_b(1)$, $\mathcal S_b(2)$, $\mathcal S_e(1)$, $\mathcal S_e(2)$, $\mathcal S(1,2,\vdash)$, $\mathcal S(1,2,\dashv)$, $\mathcal S(1,2)$, $\mathcal S(2,1)$ for the graph of Figure~\ref{fig:instance-of-mkc}. The sets of $\mathcal S(1,3)$ and $\mathcal S(2,3)$ are also represented but only their part in $A$.}
\label{fig:set-cover-2-block}
\end{figure}

For each $i \in [k]$, let us denote by $\mathcal S_b(i)$ (resp.~$\mathcal S_e(i)$), all the sets in $\mathcal S$ that contains element $x_b(i)$ (resp.~$x_e(i)$).
For each pair $i \neq j \in [k]$, we denote by $\mathcal S(i,j)$ all the sets in $\mathcal S$ that contains element $x(i,j,t)$.
Finally, for each pair $i < j \in [k]$, we denote by $\mathcal S(i,j,\vdash)$ (resp~$\mathcal S(i,j,\dashv)$) all the sets in $\mathcal S$ that contains element $y_b(i,j)$ (resp.~$y_e(i,j)$).
One can observe that the $\mathcal S_b(i)$'s, $\mathcal S_e(i)$'s, $\mathcal S(i,j)$'s, $\mathcal S(i,j,\vdash)$'s, and $\mathcal S(i,j,\dashv)$'s partition $\mathcal S$ into $k+k+k(k-1)+2{k \choose 2}=2k^2$ partite sets\footnote{We do not call them \emph{color classes} to avoid the confusion with the color classes of the instance of \mkc.}.
Thus, as each of the $2k^2$ partite sets $\mathcal S'$ has a private element which is only contained in sets of $\mathcal S'$, a solution has to contain one set in each partite set.

Assume there is a multicolored clique $\mathcal C = \{v^1_{a_1}, \ldots, v^k_{a_k}\}$ in $G$.
We show that $\mathcal T=\{S(v^i_{a_i}v^j_{a_j},v^i_{a_i}) $ $|$ $i < j \in [k]\} \cup \{S(v^i_{a_i}v^j_{a_j},v^j_{a_j}) $ $|$ $i < j \in [k]\} \cup \{S(i,a_i,\vdash) $ $|$ $i \in [k]\} \cup \{S(i,a_i,\dashv) $ $|$ $i \in [k]\}\cup \{S(v^i_{a_i}v^j_{a_j},\vdash)$ $|$ $i < j \in [k]\} \cup  \{S(v^i_{a_i}v^j_{a_j},\dashv)$ $|$ $i < j \in [k]\}$ is a set cover of $(\mathcal S,X)$ of size $2k^2$.
As $\mathcal C$ is a clique, $\mathcal T$ is well defined and it contains $2{k \choose 2}+2k+2{k \choose 2}=2k^2$ sets.
For each $i \in [k]$, the elements $x(i,1,a_i), \ldots, x(i,1,t), \ldots, x(i,k+1,1), \ldots, x(i,k+1,a_i-1)$ are covered by the sets $S(v^1_{a_1}v^i_{a_i},v^i_{a_i}), S(v^2_{a_2}v^i_{a_i},v^i_{a_i}), \ldots, S(v^i_{a_i}v^k_{a_k},v^i_{a_i})$.
Indeed, $S(v^j_{a_j}v^i_{a_i},v^i_{a_i})$ (or $S(v^i_{a_i}v^j_{a_j},v^i_{a_i})$ if $j > i$) covers all the elements $x(i,j,a_i), \ldots, x(i,j,t), x(i,j+1,1), \ldots, x(i,j+1,a_i-1)$ (again, in case $i+1 = j$, replace $\emph{j+1}$ by $\emph{i+1}$).
For each $i \in [k]$, the elements $x_b(i), x(i,1,1), \ldots, x(i,1,a_i-1), x(i,k+1,a_i), \ldots, x(i,k+1,t), x_e(i)$ and of $W(i)$ are covered by $S(i,a_i,\vdash)$ and $S(i,a_i,\dashv)$.
For all $i < j \in [k]$, say $v^i_{a_i}v^j_{a_j}$ is the $c$-th edge $e^{ij}_c$ in the arbitrary order of $E_{ij}$.
Then, the elements $y(i,j,c), y(i,j,c+1), \ldots, y(i,j,c+2|E_{ij}|-1)$ are covered by $S(v^i_{a_i}v^j_{a_j},v^i_{a_i})$ and $S(v^i_{a_i}v^j_{a_j},v^j_{a_j})$.
Finally, the elements $y_b(i,j), y(i,j,1), \ldots, y(i,j,c-1), y(i,j,$ $c+2|E_{ij}|), \ldots, y(i,j,3|E_{ij}|), y_e(i,j)$ and of $Z(i,j)$ are covered  by $S(v^i_{a_i}v^j_{a_j},\vdash)$ and $S(v^i_{a_i}v^j_{a_j},\dashv)$.

Assume now that the set-system $(X,\mathcal S)$ admits a set cover $\mathcal T$ of size $2k^2$.
As mentioned above, this solution $\mathcal T$ should contain exactly one set in each partite set (of the partition of $\mathcal S$).
For each $i \in [k]$, to cover all the elements of $W(i)$, one should take $S(i,a_i,\vdash)$ and $S(i,a'_i,\dashv)$ with $a_i \leqslant a'_i$.
Now, each set of $\mathcal S(i,j)$ has their $A$-intervals containing exactly $t$ elements.
This means that the only way of covering the $tk+2$ elements of $X(i)$ is to take $S(i,a_i,\vdash)$ and $S(i,a'_i,\dashv)$ with $a_i \geqslant a'_i$ (therefore $a_i=a'_i$), and to take all the $k-1$ sets of $\mathcal S(i,j)$ (for $j \in [k] \setminus \{i\}$) of the form $S(v^i_{a_i}v^j_{s_j},v^i_{a_i})$, for some $s_j \in [t]$.
So far, we showed that a potential solution of \ksc should stick to the same vertex $v^i_{a_i}$ in each \emph{color class}.
We now show that if one selects $S(v^i_{a_i}v^j_{s_j},v^i_{a_i})$, one should be consistent with this choice and also selects $S(v^i_{a_i}v^j_{s_j},v^j_{s_j})$.
In particular, it implies that, for each $i \in [k]$, $s_i$ should be equal to $a_i$.
For each $i \neq j \in [k]$, to cover all the elements of $Z(i,j)$, one should take $S(e^{ij}_{c_{ij}},\vdash)$ and $S(e^{ij}_{c'_{ij}},\dashv)$ with $c_{ij} \geqslant c'_{ij}$.
Now, each set of $\mathcal S(i,j)$ and each set of $\mathcal S(j,i)$ has their $B$-intervals containing exactly $|E_{ij}|$ elements.
This means that the only way of covering the $3|E_{ij}|+2$ elements of $Y(i,j)$ is to take $S(e^{ij}_{c_{ij}},\vdash)$ and $S(e^{ij}_{c'_{ij}},\dashv)$ with $c_{ij} \leqslant c'_{ij}$ (therefore, $c_{ij} = c'_{ij}$), and to take the sets $S(v^i_{a_i}v^j_{a_j},v^i_{a_i})$ and $S(v^i_{a_i}v^j_{a_j},v^j_{a_j})$.
Therefore, if there is a solution to the \ksc instance, then there is a multicolored clique $\{v^1_{a_1}, \ldots, v^k_{a_k}\}$ in $G$.

In this reduction, there is a quadratic blow-up of the parameter.
Under the ETH, it would only forbid, by Theorem~\ref{thm:chen06}, an algorithm solving \ksc on \sti instances in time $f(k)(N+M)^{o(\sqrt k)}$.
We can do the previous reduction from \textsc{Multicolored Subgraph Isomorphism} and suppress $X(i,j)$, $X(j,i)$, $Z(i,j)$, and $Y(i,j)$, and the sets defined over these elements, whenever $E_{ij}$ is empty.
One can check that the produced set cover instance is still \sti and that the way of proving correctness does not change.
Therefore, by Theorem~\ref{thm:marx10}, \ksc restricted to \sti instances cannot be solved in time $f(k)(N+M)^{o(k / \log k)}$ for any computable function $f$, unless the ETH fails.
\end{proof}

In the \textsc{2-Track Hitting Set} problem, the input consists of an integer $k$, two totally ordered ground sets $A$ and $B$ of the same cardinality, and two sets $\mathcal S_A$ of $A$-intervals, and $\mathcal S_B$ of $B$-intervals. 
In addition, the elements of $A$ and $B$ are in one-to-one correspondence $\phi: A \rightarrow B$ and each pair $(a,\phi(a))$ is called a \emph{$2$-element}. 
The goal is to find, if possible, a set $S$ of $k$ $2$-elements such that the first projection of $S$ is a hitting set of $\mathcal S_A$, and the second projection of $S$ is a hitting set of $\mathcal S_B$.

\shs is the same problem with color classes over the $2$-elements, and a restriction on the one-to-one mapping $\phi$.
Given two integers $k$ and $t$, $A$ is partitioned into $(C_1,C_2,\ldots,C_k)$ where $C_j=\{a^j_1, a^j_2, \ldots, a^j_t\}$ for each $j \in [k]$.
$A$ is ordered: $a^1_1, a^1_2, \ldots, a^1_t, a^2_1, a^2_2, \ldots, a^2_t,$ $ \ldots, a^k_1, a^k_2, \ldots, a^k_t$.
We define $C'_j:=\phi(C_j)$ and $b^j_i := \phi(a^j_i)$ for all $i \in [t]$ and $j \in [k]$.
We now impose that $\phi$ is such that, for each $j \in [k]$, the set $C'_j$ is a $B$-interval.
That is, $B$ is ordered: $C'_{\sigma(1)}, C'_{\sigma(2)}, \ldots, C'_{\sigma(k)}$ for some permutation on $[k]$, $\sigma \in \mathfrak S_k$.
For each $j \in [k]$, the order of the elements within $C'_j$ can be described by a permutation $\sigma_j \in \mathfrak S_t$ such that the ordering of $C'_j$ is: $b^j_{\sigma_j(1)}, b^j_{\sigma_j(2)}, \ldots, b^j_{\sigma_j(t)}$.
In what follows, it will be convenient to see an instance of \shs as a tuple $\mathcal I=(k \in \mathbb{N},t \in \mathbb{N}, \sigma \in \mathfrak S_k, \sigma_1 \in \mathfrak S_t, \ldots, \sigma_k \in \mathfrak S_t, \mathcal S_A, \mathcal S_B)$, where we recall that $\mathcal S_A$ is a set of $A$-intervals and $\mathcal S_B$ is a set of $B$-intervals.
The size $|\mathcal I|$ of $\mathcal I$ is defined as $kt+|\mathcal S_A|+|\mathcal S_B|$.
We denote by $[a^j_i,a^{j'}_{i'}]$ (resp. $[b^j_i,b^{j'}_{i'}]$) all the elements $a \in A$ (resp. $b \in B$) such that $a^j_i \leq_A a \leq_A a^{j'}_{i'}$ (resp. $b^j_i \leq_B b \leq_B b^{j'}_{i'}$).

Again a solution is a set of $k$ $2$-elements $\{(a_{i(1)}^1, b_{i(1)}^1), \ldots, (a_{i(k)}^k, b_{i(k)}^k)\}$, each from a distinct color class, such that $a_{i(1)}^1, \ldots, a_{i(k)}^k$ is a hitting set of $\mathcal S_A$, and $b_{i(1)}^1, \ldots, b_{i(k)}^k$ is a hitting set of $\mathcal S_B$.

\begin{figure}
\centering
\begin{tikzpicture}

\definecolor{col1}{rgb}{1,0,0}
\definecolor{col2}{rgb}{0,1,0}
\definecolor{col3}{rgb}{0,0,1}
\definecolor{col4}{rgb}{1,1,0}
\def\p{-2}
\def\e{0.5}

\def\t{6}
\def\k{4}
\foreach \j [count = \ja from 0] in {1,...,\k}{
\foreach \i in {1,...,\t}{
\node (a\j\i) at (\e * \i + \e * \ja * \t, 0) {$a^\j_\i$} ;
}
\node[draw, fill=col\j, opacity=0.3, rectangle, rounded corners, thick, fit=(a\j1) (a\j\t),label=above:$C_\j$,inner sep=-0.05cm] (c\j) {} ;
}
\node[rectangle, rounded corners, thick, fit=(c1) (c\k),label=north:$A$,inner sep=0.4cm] (A) {} ;

\node (b34) at (\e, \p) {$b^3_4$} ;
\node (b32) at (2*\e, \p) {$b^3_2$} ;
\node (b33) at (3*\e, \p) {$b^3_3$} ;
\node (b36) at (4*\e, \p) {$b^3_6$} ;
\node (b31) at (5*\e, \p) {$b^3_1$} ;
\node (b35) at (6*\e, \p) {$b^3_5$} ;
\node[draw, fill=col3, opacity=0.3, rectangle, rounded corners, thick, fit=(b34) (b35),label=below:$C'_3$,inner sep=-0.05cm] (cp3) {} ;

\begin{scope}[xshift=6*\e cm]
\node (b12) at (\e, \p) {$b^1_2$} ;
\node (b14) at (2*\e, \p) {$b^1_4$} ;
\node (b11) at (3*\e, \p) {$b^1_1$} ;
\node (b15) at (4*\e, \p) {$b^1_5$} ;
\node (b16) at (5*\e, \p) {$b^1_6$} ;
\node (b13) at (6*\e, \p) {$b^1_3$} ;
\node[draw, fill=col1, opacity=0.3, rectangle, rounded corners, thick, fit=(b12) (b13),label=below:$C'_1$,inner sep=-0.05cm] (cp1) {} ;
\end{scope}

\begin{scope}[xshift=12*\e cm]
\node (b43) at (\e, \p) {$b^4_3$} ;
\node (b46) at (2*\e, \p) {$b^4_6$} ;
\node (b45) at (3*\e, \p) {$b^4_5$} ;
\node (b42) at (4*\e, \p) {$b^4_2$} ;
\node (b41) at (5*\e, \p) {$b^4_1$} ;
\node (b44) at (6*\e, \p) {$b^1_4$} ;
\node[draw, fill=col4, opacity=0.3, rectangle, rounded corners, thick, fit=(b43) (b44),label=below:$C'_4$,inner sep=-0.05cm] (cp4) {} ;
\end{scope}

\begin{scope}[xshift=18*\e cm]
\node (b21) at (\e, \p) {$b^2_1$} ;
\node (b25) at (2*\e, \p) {$b^2_5$} ;
\node (b22) at (3*\e, \p) {$b^2_2$} ;
\node (b24) at (4*\e, \p) {$b^2_4$} ;
\node (b26) at (5*\e, \p) {$b^2_6$} ;
\node (b23) at (6*\e, \p) {$b^2_3$} ;
\node[draw, fill=col2, opacity=0.3, rectangle, rounded corners, thick, fit=(b21) (b23),label=below:$C'_2$,inner sep=-0.05cm] (cp2) {} ;
\end{scope}

\node[rectangle, rounded corners, thick, fit=(cp3) (cp2),label=south:$B$,inner sep=0.4cm] (B) {} ;

\foreach \j in {1,...,\k}{
\draw[very thick,opacity=0.8] (c\j.south) -- (cp\j.north) ; 
}

\foreach \i in {1,...,\t}{
\draw (a1\i.south) -- (b1\i.north) ; 
}

\node (s) at (6.25,-1) {$\sigma$} ;
\node (s1) at (1.5,-1) {$\sigma_1$} ;

\node (oa) at (-0.1,0) {$\leq_A:$} ;
\node (ob) at (-0.1,-2) {$\leq_B:$} ;

\foreach \i/\j/\k in {2.35/4.15/0.6, 2.85/5.15/0.9, 8.35/10.15/0.6, 5.35/6.65/-2.6}{
  \draw[thick] (\i,\k) -- (\j,\k) ;
  \draw[thick] (\i,\k+0.1) --++(0,-0.2) ;
  \draw[thick] (\j,\k+0.1) --++(0,-0.2) ;
}

\end{tikzpicture}
\caption{An illustration of a \shs instance, with $k=4$ and $t=6$.
  The permutation $\sigma \in \mathfrak S_k$ is represented with thick edges.
  Among $\sigma_1 \in \mathfrak S_t$, \dots, $\sigma_k \in \mathfrak S_t$, we only represented $\sigma_1$, for the sake of legibility.
  We also only represented four intervals of the instance, three $A$-intervals, $[a_5^1,a_2^2] = \{a_5^1,a_6^1,a_1^2,a_2^2\}$, $[a_6^1,a_4^2]$, $[a_5^3,a_2^4]$, and one $B$-interval $[b_6^1,b_3^4]=\{b_6^1,b_3^1,b_3^4\}$.}
\label{fig:permutations}
\end{figure}

We show the ETH lower bound and W[1]-hardness for \shs.
The reduction is from \ksc on \sti instances.
We transform the unions of two intervals into 2-elements, and the elements of the \ksc instance into $A$-intervals or $B$-intervals of the \shs instance.

\begin{theorem}\label{cor:main}
  \shs is \wone-hard.
  Furthermore it is not solvable in time $f(k)|\mathcal I|^{o(k / \log k)}$ for any computable function $f$, unless the ETH fails.
\end{theorem}

\begin{proof}
This result is a consequence of Lemma~\ref{thm:2-intervals-cover}.
Let $(A \uplus B, \mathcal S)$ be a hard \sti instance of \ksc, obtained from the previous reduction.
We recall that each set $S$ of $\mathcal S$ is the union of an $A$-interval with a $B$-interval: $S=S_A \uplus S_B$.
We transform each set $S$ into a $2$-element $(x_{S,A},x_{S,B})$, and each element $u$ of the \ksc instance into a set $T_u$ of the \shs instance. 
We put element $x_{S,A}$ (resp. $x_{S,B}$) into set $T_u$ whenever $u \in S \cap A=I_A$ (resp. $u \in S \cap B=I_B$). 
We call $A'$ (resp. $B'$) the set of all the elements of the form $x_{S,A}$ (resp. $x_{S,B}$).
We shall now specify an order of $A'$ and $B'$ so that the instance is \emph{structured}.
Keep in mind that elements in the \shs instance corresponds to sets in the \ksc instance.
We order the elements of $A'$ accordingly to the following ordering of the sets of the \ksc instance: $\mathcal S_b(1)$, $\mathcal S(1,2)$, $\ldots$, $\mathcal S(1,k)$, $\mathcal S_e(1)$, $\mathcal S_b(2)$, $\mathcal S(2,1)$, $\ldots$, $\mathcal S(2,k)$, $\mathcal S_e(2)$, $\ldots$, $\mathcal S_b(k)$, $\mathcal S(k,1)$, $\ldots$, $\mathcal S(k,k-1)$, $\mathcal S_e(k)$, $\mathcal S(1,2,\vdash)$, $\mathcal S(1,2,\dashv)$, $\mathcal S(1,3,\vdash)$, $\mathcal S(1,3,\dashv)$, $\ldots$, $\mathcal S(k-1,k,\vdash)$, $\mathcal S(k-1,k,\dashv)$. 
We order the elements of $B'$ accordingly to the following ordering of the sets of the \ksc instance: $\mathcal S(1,2,\vdash)$, $\mathcal S(1,2)$, $\mathcal S(2,1)$, $\mathcal S(1,2,\dashv)$, $\mathcal S(1,3,\vdash)$, $\mathcal S(1,3)$, $\mathcal S(3,1)$, $\mathcal S(1,3,\dashv)$, $\ldots$, $\mathcal S(k-1,k,\vdash)$, $\mathcal S(k-1,k)$, $\mathcal S(k,k-1)$, $\mathcal S(k-1,k,\dashv)$, $\mathcal S_b(1)$, $\mathcal S_e(1)$, $\ldots$, $\mathcal S_b(k)$, $\mathcal S_e(k)$. 
Within all those sets of sets, we order by increasing left endpoint (and then, in case of a tie, by increasing right endpoint).
One can now check that with those two orders $\leq_{A'}$ and $\leq_{B'}$, all the sets $T_u$'s are $A'$-interval or $B'$-interval.
Also, one can check that the \textsc{$2$-Track Hitting Set} instance is \emph{structured} by taking as color classes the partite sets $\mathcal S_b(i)$'s, $\mathcal S_e(i)$'s, $\mathcal S(i,j)$'s, $\mathcal S(i,j,\vdash)$'s, and $\mathcal S(i,j,\dashv)$'s.
Now, taking one $2$-element in each color class to hit all the sets $T_u$ corresponds to taking one set in each partite set of $\mathcal S$ to dominate all the elements of the \ksc instance. 
\end{proof}


$2$-track (unit) interval graphs are the intersection graphs of (unit) $2$-track intervals, where a (unit) $2$-track interval is the union of a (unit) interval in each of two parallel lines, called the first track and the second track.
A (unit) $2$-track interval may be referred to as an \emph{object}.
Two $2$-track intervals intersect if they intersect in either the first or the second track.
We observe here that many dominating problems with some geometric flavor can be restated with the terminology of $2$-track (unit) interval graphs. 

In particular, a result very close to Theorem~\ref{cor:main} was obtained recently:
\begin{theorem}[\cite{Marx15}]\label{cor:2-track-obj-int}
Given the representation of a $2$-track unit interval graph, the problem of selecting $k$ objects to dominate all the \emph{intervals} is \wone-hard, and not solvable in time $f(k)n^{o(k / \log k)}$ for any computable function $f$, unless the ETH fails.
\end{theorem}

We still had to give an \emph{alternative} proof of this result because we will need the additional property that the instance can be further assumed to have the structure depicted in Figure~\ref{fig:permutations}.
This will be crucial for showing the hardness result for \vgag. 

Other results on dominating problems in $2$-track unit interval graphs include: 
\begin{theorem}[\cite{Jiang10}]
Given the representation of a $2$-track unit interval graph, the problem of selecting $k$ objects to dominate all the objects is \wone-hard.
\end{theorem}

\begin{theorem}[\cite{Dom12}]
Given the representation of a $2$-track unit interval graph, the problem of selecting $k$ \emph{intervals} to dominate all the objects is \wone-hard.
\end{theorem}

The result of Dom et al. is formalized differently in their paper \cite{Dom12}, where the problem is defined as stabbing axis-parallel rectangles with axis-parallel lines.

\section{Parameterized hardness of the point guard variant}\label{sec:point-guard}

As exposed in Section~\ref{subsec:proof-idea}, we give a reduction from the \shs problem.
The main challenge is to design a \emph{linker} gadget that groups together specific pairs of points in the polygon.
The following introductory lemma inspires the \emph{linker} gadgets for both \pgag and \vgag.
\begin{lemma}\label{lem:linking-set-system:pg&vg}
The only minimum hitting sets of the set-system $\mathcal S = \{S_i=\{1,2,\ldots,i,$ $\overline{i+1},\overline{i+2},$ $\ldots,\overline{n}\}$ $|$ $i \in [n]\} \cup \{\overline{S}_i=\{\overline{1},\overline{2},\ldots,\overline{i},i+1,i+2,\ldots,n\}$ $|$ $i \in [n]\}$ are $\{i,\overline{i}\}$, for each $i \in [n]$.
\end{lemma}

\begin{proof}
First, for each $i \in [n]$, one may easily observe that $\{i,\overline{i}\}$ is a hitting set of $\mathcal S$.
Now, because of the sets $S_n$ and $\overline{S}_n$ one should pick one element $i$ and one element $\overline{j}$ for some $i, j \in [n]$.
If $i < j$, then set $\overline{S}_i$ is not hit, and if $i > j$, then $S_j$ is not hit.
Therefore, $i$ should be equal to $j$.
\end{proof}

Henceforth we keep this bar notation to denote pairs of homologous objects (points, vertices) that we wish to link together.  

\ParaPoint*

\begin{proof}
Given an instance $\mathcal I=(k \in \mathbb{N},t \in \mathbb{N}, \sigma \in \mathfrak S_k, \sigma_1 \in \mathfrak S_t, \ldots, \sigma_k \in \mathfrak S_t, \mathcal S_A, \mathcal S_B)$ of \shs, we build a simple polygon $\mathcal P$ with $O(kt+|\mathcal S_A|+|\mathcal S_B|)$ vertices, such that $\mathcal I$ is a YES-instance iff $\mathcal P$ can be guarded by $3k$ points.

\textbf{Outline.} We recall that $A$'s order is: $a^1_1, \ldots, a^1_t, \ldots, a^k_1, \ldots, a^k_t$ and $B$'s order is determined by $\sigma$ and the $\sigma_j$'s (see Figure~\ref{fig:permutations}).
The global strategy of the reduction is to \emph{allocate}, for each color class $j \in [k]$, $2t$ special points in the polygon $\alpha^j_1, \ldots, \alpha^j_t$ and $\beta^j_1, \ldots, \beta^j_t$.
Placing a guard in $\alpha^j_i$ (resp.~$\beta^j_i$) shall correspond to picking a $2$-element whose first (resp.~second) component is $a^j_i$ (resp.~$b^j_i$). 
The points $\alpha^j_i$'s and $\beta^j_i$'s ordered by increasing $y$-coordinates will match the order of the $a^j_i$'s along the order $\leq_A$ and then of the $b^j_i$'s along $\leq_B$.
Then, far in the horizontal direction, we will place pockets to encode each $A$-interval of $\mathcal S_A$, and each $B$-interval of $\mathcal S_B$ (see Figure~\ref{fig:interval-pg&vg}).

The critical issue will be to \emph{link} point $\alpha^j_i$ to point $\beta^j_i$.
Indeed, in the \shs problem, one selects $2$-elements (one per color class), so we should prevent one from placing two guards in $\alpha^j_i$ and $\beta^j_{i'}$ with $i \neq i'$.
The so-called \emph{point linker} gadget will be grounded in Lemma~\ref{lem:linking-set-system:pg&vg}.
Due to a technicality, we will need to introduce a \emph{copy} $\overline{\alpha}^j_i$ of each $\alpha^j_i$.
In each part of the gallery encoding a color class $j \in [k]$, the only way of guarding all the pockets with only three guards will be to place them in $\alpha^j_i$, $\overline{\alpha}^j_i$, and $\beta^j_i$ for some $i \in [t]$ (see Figure~\ref{fig:linker-pg}).
Hence, $3k$ guards will be necessary and sufficient to guard the whole $\mathcal P$ iff there is a solution to the instance of \shs.

\begin{figure}
\centering
\resizebox{300pt}{!}
{
\begin{tikzpicture}[scale=1]
\def\pockets{4}
\def\r{7}
\path[name path = p] (\r,1) -- (\r,\pockets+0.5) ;
\foreach \i in {1,...,\pockets}{
\node [vertex,label=right:$z_\i$] (z\i) at (8.5,\i) {} ;
}

\node[point,label=left:$p_1$] (p1) at (-0.3,1.1) {};
\node[point,label=left:$p_2$] (p2) at (0.2,1.6) {};
\node[point,label=left:$p_3$] (p3) at (-0.5,2) {};
\node[point,label=left:$p_4$] (p4) at (0.1,2.5) {};
\node[point,label=left:$p_5$] (p5) at (0,3.2) {};
\node[point,label=left:$p_6$] (p6) at (0.7,3.8) {};

\begin{scope}[thin,opacity=0.5]
\draw[name path global/.expanded=zpa1] (z1) -- (p1) ;
\draw[name path global/.expanded=zpb1] (z1) -- (p3) ;

\draw[name path global/.expanded=zpa2] (z2) -- (p2) ;
\draw[name path global/.expanded=zpb2] (z2) -- (p5) ;

\draw[name path global/.expanded=zpa3] (z3) -- (p4) ;
\draw[name path global/.expanded=zpb3] (z3) -- (p5) ;

\draw[name path global/.expanded=zpa4] (z4) -- (p4) ;
\draw[name path global/.expanded=zpb4] (z4) -- (p6) ;
\end{scope}

\foreach \i in {1,...,\pockets}{
\draw [name intersections={of=p and zpa\i}, very thick] (z\i) -- (intersection-1) coordinate (cza\i) {} ;
\draw [name intersections={of=p and zpb\i}, very thick] (z\i) -- (intersection-1) coordinate (czb\i) {} ; 
}

\foreach \i [count=\j from 1] in {2,...,\pockets}{
\draw[very thick] (cza\i) -- (czb\j) ;
}
\end{tikzpicture}
}
\caption{Interval gadgets encoding $\{p_1,p_2,p_3\}$, $\{p_2,p_3,p_4,p_5\}$, $\{p_4,p_5\}$, and $\{p_4,p_5,p_6\}$.}
\label{fig:interval-pg&vg}
\end{figure}

We now get into the details of the reduction.
We will introduce several characteristic lengths and compare them; when \emph{$l_1 \ll l_2$} means that $l_1$ should be thought as really small compared to $l_2$, and \emph{$l_1 \approx l_2$} means that $l_1$ and $l_2$ are roughly of the same order.
The motivation is to guide the intuition of the reader without bothering her/him too much about the details.
At the end of the construction, we will specify more concretely how those lengths are chosen. 

\textbf{Construction.}
We start by formalizing the positions of the $\alpha^j_i$'s and $\beta^j_i$'s.
We recall that we want the points $\alpha^j_i$'s and $\beta^j_i$'s ordered by increasing $y$-coordinates, to match the order of the $a^j_i$'s and $b^j_i$'s along $\leq_A$ and $\leq_B$, with first all the elements of $A$ and then all the elements of $B$.
Starting from some $y$-coordinate $y_1$ (which is the one given to point $\alpha^1_1$), the $y$-coordinates of the $\alpha^j_i$'s are regularly spaced out by an offset $y$; that is, the $y$-coordinate of $\alpha^j_i$ is $y_1+(i+(j-1)t)y$.
Between the $y$-coordinate of the last element in $A$ (i.e., $a^k_t$ whose $y$-coordinate is $y_1+(kt-1)y$) and the first element in $B$, there is a large offset $L$, such that the $y$-coordinate of $\beta^j_i$ is $y_1+(kt-1)y+L+(\text{ind}(b^j_i)-1)y$ (for any $j \in [k]$ and $i \in [t]$) where $\text{ind}(b^j_i)$ is the \emph{index} of $b^j_i$ along the order $\leq_B$, that is the number of $b \in B$ such that $b \leq_B b^j_i$.

For each color class $j \in [k]$, let $x_j := x_1 + (j-1)D$ for some $x$-coordinate $x_1$ and value $D$, and $y_j := y_1 + (j-1)ty$.
The allocated points $\alpha^j_1, \alpha^j_2, \alpha^j_3, \ldots, \alpha^j_t$ are on a line at coordinates: $(x_j,y_j), (x_j+x,y_j+y), (x_j+2x,y_j+2y), \ldots, (x_j+(t-1)x,y_j+(t-1)y)$, for some value $x$.
We place, to the left of those points, a rectangular pocket $\mathcal P_{j,r}$ of width, say, $y$ and length, say\footnote{the exact width and length of this pocket are not relevant; the reader may just think of $\mathcal P_{j,r}$ as a thin pocket which forces to place a guard on a thin strip whose uppermost boundary is $\ell(\alpha^j_1,\alpha^j_t)$}, $tx$ such that the uppermost longer side of the rectangular pocket lies on the line $\ell(\alpha^j_1,\alpha^j_t)$ (see Figure~\ref{fig:weak-linker-pg}).
The $y$-coordinates of $\beta^j_1, \beta^j_2, \beta^j_3, \ldots, \beta^j_t$ have already been defined.
We set, for each $i \in [t]$, the $x$-coordinate of $\beta^j_i$ to $x_j+(i-1)x$, so that $\beta^j_i$ and $\alpha^j_i$ share the same $x$-coordinate. 
One can check that it is consistent with the previous paragraph.
We also observe that, by the choice of the $y$-coordinate for the $\beta^j_i$'s, we have both encoded the permutations $\sigma_j$'s and permutation $\sigma$ (see Figure~\ref{fig:overall-pg} or Figure~\ref{fig:weak-linker-pg}).

Our construction almost exclusively rely on so-called \emph{triangular pockets}.
Henceforth, for a vertex $v$ and two points $p$ and $p'$, we call \emph{a triangular pocket rooted at vertex $v$ and supported by $\ray(v,p)$ and $\ray(v,p')$} a sub-polygon $w, v, w'$ (a triangle) such that $\ray(v,w)$ passes through $p$, $\ray(v,w')$ passes through $p'$, while $w$ and $w'$ are close to $v$ (sufficiently close not to interfere with the rest of the construction).
We say that $v$ is the \emph{root} of the triangular pocket, that we often denote by $\mathcal P(v)$.
We also say that the pocket $\mathcal P(v)$ \emph{points} towards $p$ and $p'$.

We now encode the $A$-intervals and $B$-intervals with triangular pockets.
At the $x$-coordinate $x_k+(t-1)x+F$, for some large value $F$, we put between $y$-coordinates $y_1$ and $y_k+(kt-1)y$, for each $A$-interval $I_q=[a^j_i,a^{j'}_{i'}] \in \mathcal S_A$ we put one triangular pocket $\mathcal P(z_{A,q})$ rooted at vertex $z_{A,q}$ and supported by $\ray(z_{A,q},\alpha^j_i)$ and $\ray(z_{A,q},\alpha^{j'}_{i'})$.
Intuitively, if $y \ll x \ll D \ll F$, the only $\alpha^{j''}_{i''}$ seeing vertex $z_{A,q}$ should be all the points such that $a^j_i \leq_A a^{j''}_{i''} \leq_A a^{j'}_{i'}$ (see Figure~\ref{fig:overall-pg} and Figure~\ref{fig:interval-pg&vg}).
We place those $|\mathcal S_A|$ pockets along the $y$-axis, and space them out by distance $s$.
To guarantee that we have enough room to place all those pockets, $s \ll y$ shall later hold.
Similarly, we place at the same $x$-coordinate $x_k+(t-1)x+F$ each of the $|\mathcal S_B|$ triangular pockets $\mathcal P(z_{B,q})$ rooted at vertex $z_{B,q}$ and supported by $\ray(z_{B,q},\beta^j_i)$ and $\ray(z_{B,q},\beta^{j'}_{i'})$ for $B$-interval $[b^j_i,b^{j'}_{i'}] \in \mathcal S_B$; and we space out those pockets by distance $s$ along the $y$-axis between $x$-coordinates $y_1+(kt-1)y+L$ and $y_1+2(kt-1)y+L$.
We do not specify an order to the $z_{A,q}$'s (resp.~the $z_{B,q}$'s) along the $y$-axis since we do not need that to prove the reduction correct.
The different values ($s$, $x$, $y$, $D$, $L$, and $F$) introduced so far compare in the following way: $s \ll y \ll x \ll D \ll F$, and $x \ll L \ll F$ (see Figure~\ref{fig:overall-pg}).

We now describe the \emph{linker gadget}, or how to force consistent pairs of guards $\alpha^j_i$ and its \asso $\beta^j_i$.
The idea is that pairs of guards $\alpha^j_i, \beta^j_i$ will be very effective since the two points see disjoint sets of pockets, whereas pairs $\alpha^j_i, \beta^j_{i'}$ (with $i \neq i'$) will overlap on some pockets, and miss some other pockets completely.

For each $j \in [k]$, let us mentally draw $\ray(\alpha^j_t,\beta^j_1)$ and consider points slightly to the left of this ray at a distance, say, $L'$ from point $\alpha^j_t$.
Let us call $\mathcal R^j_{\text{left}}$ that informal region of points.
Any point in $\mathcal R^j_{\text{left}}$ sees, from right to left, in this order $\alpha^j_1$, $\alpha^j_2$ up to $\alpha^j_t$, and then, $\beta^j_1$, $\beta^j_2$ up to $\beta^j_t$.
This observation relies on the fact that $y \ll x \ll L$.
So, from the distance, the points $\beta^j_1, \ldots, \beta^j_t$ look almost \emph{flat}.
It makes the following construction possible.
In $\mathcal R^j_{\text{left}}$, for each $i \in [t-1]$, we place a triangular pocket $\mathcal P(c^j_i)$ rooted at vertex $c^j_i$ and supported by $\ray(c^j_i,\alpha^j_{i+1})$ and $\ray(c^j_i,\beta^j_i)$. 
We place also a triangular pocket $\mathcal P(c^j_t)$ rooted at $c^j_t$ supported by $\ray(c^j_t,\beta^j_1)$ and $\ray(c^j_t,\beta^j_t)$.
We place the vertices $c^j_i$ ($i \in [t]$) at the same $y$-coordinate and we space them out by distance $x$ along the $x$-axis (see Figure~\ref{fig:weak-linker-pg}).
Similarly, let us informally refer to the region slightly to the right of $\ray(\alpha^j_1,\beta^j_t)$ at a distance $L'$ from point $\alpha^j_1$, as $\mathcal R^j_{\text{right}}$.
Any point $\mathcal R^j_{\text{right}}$ sees, from right to left, in this order $\beta^j_1$, $\beta^j_2$ up to $\beta^j_t$, and then, $\alpha^j_1$, $\alpha^j_2$ up to $\alpha^j_t$.
Therefore, one can place in $\mathcal R^j_{\text{left}}$, for each $i \in [t-1]$, a triangular pocket $\mathcal P(d^j_i)$ rooted at $d^j_i$ supported by $\ray(d^j_i,\beta^j_{i+1})$ and $\ray(c^j_i,\alpha^j_i)$.
We place also a triangular pocket $\mathcal P(d^j_t)$ rooted at $d^j_t$ supported by $\ray(d^j_t,\alpha^j_1)$ and $\ray(d^j_t,\alpha^j_t)$.
Again, those $t$ pockets can be put at the same $y$-coordinate and spaced out horizontally by $x$ (see Figure~\ref{fig:weak-linker-pg}).
We denote by $\mathcal P_{j,\alpha,\beta}$ the set of pockets $\{\mathcal P(c^j_1), \ldots, \mathcal P(c^j_t), \mathcal P(d^j_1), \ldots, \mathcal P(d^j_t)\}$ and informally call it the \emph{weak point linker} (or simply, \emph{weak linker}) of $\alpha^j_1, \ldots, \alpha^j_t$ and $\beta^j_1, \ldots, \beta^j_t$.
We may call the pockets of $\mathcal R^j_{\text{left}}$ (resp. $\mathcal R^j_{\text{right}}$) \emph{left} pockets (resp. \emph{right} pockets).

\begin{figure}
\centering
\resizebox{350pt}{!}
{
\begin{tikzpicture}[scale=0.35]
\foreach \i in {1,...,6}{
\node[point,label=below:$\alpha_\i$] (a\i) at (\i,0.1 * \i) {} ;
}

\node[point,label=above:$\beta_1$] (b1) at (1,8.4) {} ;
\node[point,label=above:$\beta_2$] (b2) at (2,8.7) {} ;
\node[point,label=above:$\beta_3$] (b3) at (3,8.1) {} ;
\node[point,label=above:$\beta_4$] (b4) at (4,8.2) {} ;
\node[point,label=above:$\beta_5$] (b5) at (5,8.6) {} ;
\node[point,label=above:$\beta_6$] (b6) at (6,8.3) {} ;

\foreach \i in {1,...,6}{
\coordinate (cb\i) at (b\i) ;
}

\path [name path = p] (-16,21.5) -- (23,21.5) ;
\foreach \i in {1,...,6}{
\node [vertex,label=above:$d_\i$] (c\i) at (17+\i,23) {} ;
\draw [name path global/.expanded=e\i,very thin] (c\i) -- (a\i) ;
\node [vertex,label=above:$c_\i$] (d\i) at (-17+\i,23) {} ;
\draw [name path global/.expanded=g\i,very thin] (d\i) -- ($(b\i)!-0.7cm!(d\i)$) coordinate (q\i) ;
}

\foreach \i [count=\j from 2] in {1,...,5}{
\draw[name path global/.expanded=f\i,very thin] (c\i) -- ($(b\j)!-0.7cm!(c\i)$) coordinate (p\i) ;       
}

\draw[name path = f6, very thin] (c6) -- (a1) ;  


\draw[name path = h6, very thin] (d6) -- ($(b1)!-0.7cm!(d6)$) coordinate (p6) ;

\foreach \i [count=\j from 2] in {1,...,5}{
\draw[name path global/.expanded=h\i, very thin] (d\i) -- (a\j) ;       
}

\foreach \i in {1,...,6}{
\draw [name intersections={of=p and e\i}, very thick] (c\i) -- (intersection-1) coordinate (s\i) {} ; 
\draw [name intersections={of=p and f\i}, very thick] (c\i) -- (intersection-1) coordinate (t\i) {} ; 
\draw [name intersections={of=p and g\i}, very thick] (d\i) -- (intersection-1) coordinate (u\i) {} ; 
\draw [name intersections={of=p and h\i}, very thick] (d\i) -- (intersection-1) coordinate (v\i) {} ; 
}

\foreach \i [count=\j from 2] in {1,...,5}{
\draw[very thick] (s\i) -- (t\j) ; 
\draw[very thick] (u\i) -- (v\j) ;       
}

\draw[very thick] (u6) -- (t1) ;

\begin{scope}[opacity=0.5]
\foreach \i [count = \j from 2] in {1,...,5}{
\fill (cb\j) -- (p\i) -- (q\j) ;
}
\fill (cb1) -- (p6) -- (q1) ;
\end{scope}

\begin{scope}[xshift=-9cm,yshift=-0.9cm]
\draw[very thick] (6,0.5) -- (1,0) -- (1,0.2) -- (6,0.7)  ;
\end{scope}

\end{tikzpicture}
}
\caption{Weak point linker gadget $\mathcal P_{j,\alpha,\beta}$ with $t=6$. We omit the superscript $j$ in all the labels.}
\label{fig:weak-linker-pg}
\end{figure}

As we will show later, if one wants to guard with only two points all the pockets of $\mathcal P_{j,\alpha,\beta} = \{\mathcal P(c^j_1), \ldots, \mathcal P(c^j_t), \mathcal P(d^j_1), \ldots, \mathcal P(d^j_t)\}$ and one first decides to put a guard on point $\alpha^j_i$ (for some $i \in [t]$), then one is not forced to put the other guard on point $\beta^j_i$ but only on an area whose uppermost point is $\beta^j_i$ (see the shaded areas below the $b^j_i$'s in Figure~\ref{fig:weak-linker-pg}).
Now, if $\beta^j_1, \ldots, \beta^j_t$ would all lie on a same line $\ell$, we could shrink the shaded area of each $\beta^j_i$ (Figure~\ref{fig:weak-linker-pg}) down to the single point $\beta^j_i$ by adding a thin rectangular pocket on $\ell$ (similarly to what we have for $\alpha^j_1, \ldots, \alpha^j_t$).
Naturally we need that $\beta^j_1, \ldots, \beta^j_t$ are \emph{not} on the same line, in order to encode $\sigma_j$.

The remedy we suggest is to make a triangle of weak linkers.
For each $j \in [k]$, we allocate $t$ points $\overline{\alpha}^j_1, \overline{\alpha}^j_2, \ldots, \overline{\alpha}^j_t$ on a horizontal line, spaced out by distance $x$, say, $\approx \frac{D}{2}$ to the right and $\approx L$ to the up of $\beta^j_t$.
We put a thin horizontal rectangular pocket $\mathcal P_{j,\overline{r}}$ of the same dimension as $\mathcal P_{j,r}$ such that the lowermost longer side of $\mathcal P_{j,\overline{r}}$ is on the line $\ell(\overline{\alpha}^j_1,\overline{\alpha}^j_t)$.
We add the $2t$ pockets corresponding to a weak linker $\mathcal P_{j,\alpha,\overline{\alpha}}$ between $\alpha^j_1, \ldots, \alpha^j_t$ and $\overline{\alpha}^j_1, \ldots, \overline{\alpha}^j_t$ as well as the $2t$ pockets of a weak linker $\mathcal P_{j,\overline{\alpha},\beta}$ between $\overline{\alpha}^j_1, \ldots, \overline{\alpha}^j_t$ and $\beta^j_1, \ldots, \beta^j_t$ as pictured in Figure~\ref{fig:linker-pg}.
We denote by $\mathcal P_j$ the union $\mathcal P_{j,r} \cup \mathcal P_{j,\overline{r}} \cup \mathcal P_{j,\alpha,\beta} \cup \mathcal P_{j,\alpha,\overline{\alpha}} \cup \mathcal P_{j,\overline{\alpha},\beta}$ of all the pockets involved in the encoding of color class $j$.
Now, say, one wants to guard all the pockets of $\mathcal P_j$ with only three points, and chooses to put a guard on $\alpha^j_i$ (for some $i \in [t]$).
Because of the pockets of $\mathcal P_{j,\alpha,\overline{\alpha}} \cup P_{j,\overline{r}}$, one is forced to place a second guard precisely on $\overline{\alpha}^j_i$.
Now, because of the weak linker $\mathcal P_{j,\alpha,\beta}$ the third guard should be on a region whose uppermost point is $\beta^j_i$, while, because of $\mathcal P_{j,\overline{\alpha},\beta}$ the third guard should be on a region whose lowermost point is $\beta^j_i$.
The conclusion is that the third guard should be put precisely on $\beta^j_i$. 
This \emph{triangle} of weak linkers is called the \emph{linker} of color class $j$. 
The $k$ linkers are placed accordingly to Figure~\ref{fig:overall-pg}.
This ends the construction.

\begin{figure}[h!]
\centering
\resizebox{450pt}{!}
{
\begin{tikzpicture}[scale=0.15]
\foreach \i in {1,...,6}{
\node[spoint] (a\i) at (\i,0.1 * \i-12) {} ;
}

\begin{scope}[xshift=35cm,yshift=30cm]
\foreach \i in {1,...,6}{
\node[spoint] (zb\i) at (\i-6,-3) {} ;
\coordinate (zc\i) at (11+1.5 * \i,10) {} ;
\coordinate (zd\i) at (-3+\i,10) {} ;
}
\end{scope}

\foreach \i in {1,...,6}{
\coordinate (yc\i) at (-35,-3 - 0.7 * \i) {} ;
\coordinate (yd\i) at (-35,10 -0.7 * \i) {} ;
}

\path [name path = zp] (0,38) -- (70,38) ;
\path [name path = yp]  (-33,-10) -- (-33,20) ;

\begin{scope}[yshift=8cm]
\node[spoint] (b1) at (1,8.4) {} ;
\node[spoint] (b2) at (2,8.7) {} ;
\node[spoint] (b3) at (3,8.1) {} ;
\node[spoint] (b4) at (4,8.2) {} ;
\node[spoint] (b5) at (5,8.6) {} ;
\node[spoint] (b6) at (6,8.3) {} ;
\end{scope}


\path [name path = p] (-16,28.5) -- (23,28.5) ;

\begin{scope}[yshift=7cm]
\foreach \i in {1,...,6}{
\coordinate (c\i) at (17+\i,23) {} ;
\draw [name path global/.expanded=e\i,green,very thin] (c\i) -- (a\i) ;
\coordinate (d\i) at (-17+\i,23) {} ;
\draw [name path global/.expanded=g\i,green,very thin] (d\i) -- ($(b\i)!-0.7cm!(d\i)$) coordinate (q\i) ;
}
\end{scope}

\foreach \i [count=\j from 2] in {1,...,5}{
\draw[name path global/.expanded=f\i,green,very thin] (c\i) -- ($(b\j)!-0.7cm!(c\i)$) coordinate (p\i) ;       
}

\draw[name path = f6, green,very thin] (c6) -- (a1) ;  

\draw[name path = h6, green,very thin] (d6) -- ($(b1)!-0.7cm!(d6)$) coordinate (p6) ;

\foreach \i [count=\j from 2] in {1,...,5}{
\draw[name path global/.expanded=h\i, green,very thin] (d\i) -- (a\j) ;       
}

\foreach \i in {1,...,6}{
\draw [name intersections={of=p and e\i}, very thick] (c\i) -- (intersection-1) coordinate (s\i) {} ; 
\draw [name intersections={of=p and f\i}, very thick] (c\i) -- (intersection-1) coordinate (t\i) {} ; 
\draw [name intersections={of=p and g\i}, very thick] (d\i) -- (intersection-1) coordinate (u\i) {} ; 
\draw [name intersections={of=p and h\i}, very thick] (d\i) -- (intersection-1) coordinate (v\i) {} ; 
}

\foreach \i [count=\j from 2] in {1,...,5}{
\draw[very thick] (s\i) -- (t\j) ; 
\draw[very thick] (u\i) -- (v\j) ;       
}

\draw[very thick] (u6) -- (t1) ;


\foreach \i in {1,...,6}{
\draw [name path global/.expanded=ze\i,blue,very thin] (zc\i) -- (a\i) ;
\draw [name path global/.expanded=zg\i,blue,very thin] (zd\i) -- (zb\i) coordinate (zq\i) ;
}

\foreach \i [count=\j from 2] in {1,...,5}{
\draw[name path global/.expanded=zf\i,blue,very thin] (zc\i) -- (zb\j) coordinate (zp\i) ;       
}

\draw[name path = zf6, blue,very thin] (zc6) -- (a1) ;  

\draw[name path = zh6, blue,very thin] (zd6) -- (zb1) coordinate (zp6) ;

\foreach \i [count=\j from 2] in {1,...,5}{
\draw[name path global/.expanded=zh\i, blue,very thin] (zd\i) -- (a\j) ;       
}

\foreach \i in {1,...,6}{
\draw [name intersections={of=zp and ze\i}, very thick] (zc\i) -- (intersection-1) coordinate (zs\i) {} ; 
\draw [name intersections={of=zp and zf\i}, very thick] (zc\i) -- (intersection-1) coordinate (zt\i) {} ; 
\draw [name intersections={of=zp and zg\i}, very thick] (zd\i) -- (intersection-1) coordinate (zu\i) {} ; 
\draw [name intersections={of=zp and zh\i}, very thick] (zd\i) -- (intersection-1) coordinate (zv\i) {} ; 
}

\foreach \i [count=\j from 2] in {1,...,5}{
\draw[very thick] (zs\i) -- (zt\j) ; 
\draw[very thick] (zu\i) -- (zv\j) ;       
}

\draw[very thick] (zu6) -- (zt1) ;

\draw[very thick] (zv1) -- (s6) ;


\foreach \i in {1,...,6}{
\draw [name path global/.expanded=ye\i,red, very thin] (yc\i) -- (zb\i) ;
\draw [name path global/.expanded=yg\i,red, very thin] (yd\i) -- ($(b\i)!-0.7cm!(yd\i)$) coordinate (yq\i) ;
}

\foreach \i [count=\j from 2] in {1,...,5}{
\draw[name path global/.expanded=yf\i,red, very thin] (yc\i) -- ($(b\j)!-0.7cm!(yc\i)$) coordinate (yp\i) ;       
}

\draw[name path = yf6, red, very thin] (yc6) -- (zb1) ;  

\draw[name path = yh6, red, very thin] (yd6) -- ($(b1)!-0.7cm!(yd6)$) coordinate (yp6) ;

\foreach \i [count=\j from 2] in {1,...,5}{
\draw[name path global/.expanded=yh\i, red, very thin] (yd\i) -- (zb\j) ;       
}

\foreach \i in {1,...,6}{
\draw [name intersections={of=yp and ye\i}, very thick] (yc\i) -- (intersection-1) coordinate (ys\i) {} ; 
\draw [name intersections={of=yp and yf\i}, very thick] (yc\i) -- (intersection-1) coordinate (yt\i) {} ; 
\draw [name intersections={of=yp and yg\i}, very thick] (yd\i) -- (intersection-1) coordinate (yu\i) {} ; 
\draw [name intersections={of=yp and yh\i}, very thick] (yd\i) -- (intersection-1) coordinate (yv\i) {} ; 
}

\foreach \i [count=\j from 2] in {1,...,5}{
\draw[very thick] (ys\i) -- (yt\j) ; 
\draw[very thick] (yu\i) -- (yv\j) ;       
}

\draw[very thick] (yu6) -- (yt1) ;
\begin{scope}[yshift=2cm]
\draw[very thick] (yv1) -- (-20,25) -- (-26,25) -- (-26,25.5) -- (-20,25.5) -- (v1) ;
\end{scope}

\begin{scope}[xshift=-20cm,yshift=-13.5cm]
\draw[very thick] (6,0.4) -- (1,-0.1) -- (1,0.3) -- (6,0.8) -- (ys6) ;
\end{scope}

\node at (2,26) {$\mathcal P_{j,\alpha,\beta}$} ;
\node at (42,40) {$\mathcal P_{j,\alpha,\overline{\alpha}}$} ;
\node at (-29,3) {$\mathcal P_{j,\overline{\alpha},\beta}$} ;

\end{tikzpicture}
}
\caption{Point linker gadget $\mathcal P_j$: a triangle of (three) weak point linkers $\mathcal P_{j,\alpha,\beta}$, $\mathcal P_{j,\alpha,\overline{\alpha}}$, $\mathcal P_{j,\overline{\alpha},\beta}$, and two rectangular pockets forcing one guard on the lines $\ell(\alpha_1^j,\alpha_2^j)=\ell(\alpha_1^j,\alpha_t^j)$ and $\ell(\overline{\alpha}_1^j,\overline{\alpha}_2^j)=\ell(\overline{\alpha}_1^j,\overline{\alpha}_t^j)$.}
\label{fig:linker-pg}
\end{figure}

\textbf{Specification of the distances.}
We can specify the coordinates of positions of all the vertices by fractions of integers. 
These integers are polynomially bounded in $n$. If we want to get integer coordinates, we can transform the rational coordinates to integer coordinates by multiplying all of them with the least common multiple of all the denominators, which is not polynomially bounded anymore. The length of the integers in binary is still polynomially bounded.

We can safely set $s$ to one, as it is the smallest length, we specified.
We will put $|\mathcal{S}_A|$ pockets on track $1$ and $|\mathcal{S}_B|$ pockets on track $2$. 
It is sufficient to have an opening space of one between them.
Thus, the space on the right side of $\mathcal{P}$, for all pockets of track $1$ is bounded by $2\cdot |\mathcal{S}_A|$.
Thus setting $y$ to $|\mathcal{S}_A| + |\mathcal{S}_B|$ secures us that we have plenty of space to place all the pockets.
We specify $F = (|\mathcal{S}_A|+|\mathcal{S}_B|) D k = y\cdot D \cdot k$.
We have to show that this is large enough to guarantee that the pockets on track $1$ distinguish the picked points only by the $y$-coordinate.
Let $p$ and $q$ be two points among the $\alpha_{i}^j$. Their vertical distance is upper bounded by $D k$ and their horizontal distance is lower bounded by $y$. Thus the slope of $\ell = \ell(p,q)$ is at least $\tfrac{y}{Dk}$. At the right side of $\mathcal{P}$ the line $\ell$ will be at least $F \tfrac{y}{Dk}$ above the pockets of track $1$. 
Note $F \tfrac{y}{Dk} = y D k \cdot \tfrac{y}{Dk}  > y^2 > |\mathcal{S}_A|^2 > 2\cdot |\mathcal{S}_A|$.
The same argument shows that $F$ is sufficiently large for track $2$. 

The remaining lengths $x, L, L'$, and $D$ can be specified in a similar fashion.
For the construction of the pockets, let $s\in\mathcal{S}_A$ be an $A$-interval with endpoints $a$ and $b$, represented by some points $p$ and $q$ and assume the opening vertices $v$ and $w$ of the triangular pocket are already specified.
Then the two lines $\ell(p,v)$ and $\ell(q,w)$ will meet at some point $x$ to the right of $v$ and $w$. 
By Lemma~\ref{lem:Rational}, $x$ has rational coordinates and the integers to represent them can be expressed by the coordinates of $p, q, v,$ and $w$.
This way, all the pockets can be explicitly constructed using rational coordinates as claimed above.
\begin{figure}
	\includegraphics[page = 2]{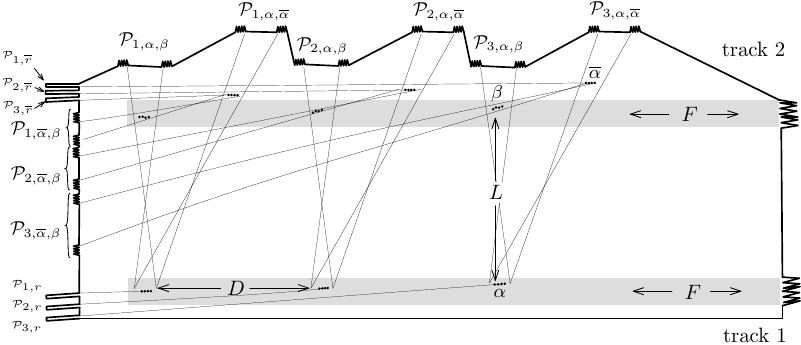}
	\caption{The overall picture of the reduction with $k=3$.
		The combination of $\mathcal P_{j,\alpha,\beta}$, $\mathcal P_{j,\alpha,\overline{\alpha}}$, $\mathcal P_{j,\overline{\alpha},\beta}$, $P_{j,r}$, and $P_{j,\overline{r}}$ forces to place pairs of guards at $\alpha_{i(j)}^j, \beta_{i(j)}^j$, analogously to the \shs semantics.
		The $y$-coordinates of these points encode the total orders over $A$ and $B$.
		The $A$-intervals are encoded by triangular pockets in track 1, while the $B$-intervals are encoded in track 2.}
	\label{fig:overall-pg}
\end{figure}

\textbf{Correctness.}
We now show that the reduction is correct.
The following lemma is the main argument for the easier implication: \emph{if $\mathcal I$ is a YES-instance, then the gallery that we build can be guarded with $3k$ points}. 

\begin{lemma}\label{positive:pg}
$\forall j \in [k]$, $\forall i \in [t]$, the three \asso points $\alpha^j_i$, $\overline{\alpha}^j_i$, $\beta^j_i$ guard $\mathcal P_j$ entirely.
\end{lemma}

{
\begin{proof}
The rectangular pockets $\mathcal P_{j,r}$ and $\mathcal P_{j,\overline{r}}$ are entirely seen by $\alpha^j_i$ and $\overline{\alpha}^j_i$, respectively.
The pockets $\mathcal P(c^j_1), \mathcal P(c^j_2), \ldots \mathcal P(c^j_{i-1})$ and $\mathcal P(d^j_i), \mathcal P(d^j_{i+1}), \ldots \mathcal P(d^j_t)$ are all entirely seen by $\alpha^j_i$, while the pockets $\mathcal P(c^j_i), \mathcal P(c^j_{i+1}), \ldots \mathcal P(c^j_t)$ and $\mathcal P(d^j_1), \mathcal P(d^j_2), \ldots \mathcal P(d^j_{i-1})$ are all entirely seen by $\beta^j_i$.
This means that $\alpha^j_i$ and $\beta^j_i$ jointly see all the pockets of $\mathcal P_{j,\alpha,\beta}$.
Similarly, $\alpha^j_i$ and $\overline{\alpha}^j_i$ jointly see all the pockets of $\mathcal P_{j,\alpha,\overline{\alpha}}$, and $\overline{\alpha}^j_i$ and $\beta^j_i$ jointly see all the pockets of $\mathcal P_{j,\overline{\alpha},\beta}$.
Therefore, $\alpha^j_i$, $\overline{\alpha}^j_i$, $\beta^j_i$ jointly see all the pockets of $\mathcal P_j$.
\end{proof}
}

Assume that $\mathcal I$ is a YES-instance and let $\{(a^1_{s_1},b^1_{s_1}), \ldots, (a^k_{s_k},b^k_{s_k})\}$ be a solution.
We claim that $G=\{\alpha^1_{s_1},\overline{\alpha}^1_{s_1},\beta^1_{s_1}, \ldots, \alpha^k_{s_k},\overline{\alpha}^k_{s_k},\beta^k_{s_k}\}$ guard the whole polygon $\mathcal P$.
By Lemma~\ref{positive:pg}, $\forall j \in [k]$, $\mathcal P_j$ is guarded.
For each $A$-interval (resp.~$B$-interval) in $\mathcal S_A$ (resp.~$\mathcal S_B$) there is at least one $2$-element $(a^j_{s_j},b^j_{s_j})$ such that $a^j_{s_j} \in \mathcal S_A$ (resp. $b^j_{s_j} \in \mathcal S_B$).
Thus, the corresponding pocket is guarded by $\alpha^j_{s_j}$ (resp.~$\beta^j_{s_j}$).
The rest of the polygon $\mathcal P$ (which is not part of pockets) is guarded by, for instance, $\{\overline{\alpha}^1_{s_1}, \ldots, \overline{\alpha}^k_{s_k}\}$.
So, $G$ is indeed a solution and it contains $3k$ points.

We now assume that there is a set $G$ of $3k$ points guarding $\mathcal P$.
We will then show that $\mathcal I$ is a YES-instance.
We observe that no point of $\mathcal P$ sees inside two triangular pockets one being in $\mathcal P_{j,\alpha,\gamma}$ and the other in $\mathcal P_{j',\alpha,\gamma'}$ with $j \neq j'$ and $\gamma, \gamma' \in \{\beta,\overline{\alpha}\}$.
Further, $V(r(\mathcal P_{j,\alpha,\beta} \cup \mathcal P_{j,\alpha,\overline{\alpha}})) \cap V(r(\mathcal P_{j',\alpha,\beta} \cup \mathcal P_{j',\alpha,\overline{\alpha}})) = \emptyset$ when $j \neq j'$, where $r$ maps a set of triangular pockets to the set of their root.
Also, for each $j \in [k]$, seeing $\mathcal P_{j,\alpha,\beta}$ and $\mathcal P_{j,\alpha,\overline{\alpha}}$ entirely requires at least $3$ points.
This means that for each $j \in [k]$, one should place three guards in $V(r(\mathcal P_{j,\alpha,\beta} \cup \mathcal P_{j,\alpha,\overline{\alpha}}))$.
Furthermore, one can observe that, among those three points, one should guard a triangular pocket $\mathcal P_{j',r}$ and another should guard $\mathcal P_{j'',\overline{r}}$.
Thus a set $S_1$, consisting of three guards of $G$, sees $\mathcal P_1$ and two rectangular pockets $\mathcal P_{j',r}$ and $\mathcal P_{j'',\overline{r}}$.

Let us call $\ell_1$ (resp. $\ell'_1$) the line corresponding to the extension of the uppermost (resp. lowermost) longer side of $\mathcal P_{1,r}$ (resp. $\mathcal P_{1,\overline{r}}$).
The only points of $\mathcal P$ that can see a rectangular pocket $\mathcal P_{j',r}$ and at least $t$ pockets of $\mathcal P_{1,\alpha,\overline{\alpha}}$ are on $\ell_1$: more specifically, they are the points $\alpha^1_1, \ldots, \alpha^1_t$.
 The only points that can see a rectangular pocket $\mathcal P_{j'',\overline{r}}$ and at least $t$ pockets of $\mathcal P_{1,\alpha,\overline{\alpha}}$ are on $\ell'_1$: they are the points $\overline{\alpha}^1_1, \ldots, \overline{\alpha}^1_t$.
As $\mathcal P_{1,\alpha,\overline{\alpha}}$ has $2t$ pockets, $S_1$ should contain two points $\alpha^1_i$ and $\overline{\alpha}^1_{i'}$.
By the argument of Lemma~\ref{lem:linking-set-system:pg&vg}, $i$ should be equal to $i'$ (otherwise, $i<i'$ and the left pocket pointing towards $\overline{\alpha}^1_{i'-1}$ and $\alpha^1_{i'}$ is not seen, or $i>i'$ and the right pocket pointing towards $\alpha^1_{i+1}$ and $\overline{\alpha}^1_i$ is not seen).
We denote by $s_1$ this shared value.
Now, to see the left pocket $\mathcal P(c^1_{s_1})$ and the right pocket $\mathcal P(d^1_{s_1-1})$ (that should still be seen), the third guard should be to the left of $\ell(c^1_{s_1},\beta^1_{s_1})$ and to the right of $\ell(d^1_{s_1-1},\beta^1_{s_1})$ (see shaded area of Figure~\ref{fig:weak-linker-pg}).
That is, the third guard of $S_1$ should be on a region in which $\beta^1_{s_1}$ is the uppermost point.
The same argument with the pockets of $\mathcal P_{1,\overline{\alpha},\beta}$ implies that the third guard should also be on a region in which $\beta^1_{s_1}$ is the lowermost point.
Thus, the third guard of $S_1$ has to be the point~$\beta^1_{s_1}$.
Therefore $S_1 = \{\alpha^1_{s_1}, \overline{\alpha}^1_{s_1}, \beta^1_{s_1}\}$, for some $s_1 \in [t]$.

As none of those three points see any pocket $\mathcal P_{j,\overline{\alpha},\beta}$ with $j>1$ (we already mentioned that no pocket of $\mathcal P_{j,\alpha,\beta}$ and $\mathcal P_{j,\alpha,\overline{\alpha}}$ with $j>1$ can be seen by those points), we can repeat the argument for the second color class; and so forth up to color class $k$. 
Thus, $G$ is of the form $\{\alpha^1_{s_1},\overline{\alpha}^1_{s_1},\beta^1_{s_1}, \ldots, \alpha^k_{s_k},\overline{\alpha}^k_{s_k},\beta^k_{s_k}\}$.
As $G$ also guards all the pockets of tracks 1 and 2, the set of $k$ 2-elements $\{(a^1_{s_1},b^1_{s_1}), \ldots, (a^k_{s_k},b^k_{s_k})\}$ hits all the $A$-intervals of $\mathcal S_A$, and the $B$-intervals of $\mathcal S_B$.
\end{proof}


\section{Parameterized hardness of the vertex guard variant}\label{sec:vertex-guard}

We now turn to the vertex guard variant and show the same hardness result.
Again, we reduce from \shs and our main task is to design a \emph{linker gadget}.
Though, \emph{linking} pairs of vertices turns out to be very different from \emph{linking} pairs of points.
Therefore, we have to come up with fresh ideas to carry out the reduction.
In a nutshell, the principal ingredient is to \emph{link} pairs of convex vertices by introducing reflex vertices at strategic places.
As placing guards on those reflex vertices is not supposed to happen in the \shs instance, we design a so-called \emph{filter gadget} to prevent any solution from doing so.   

 \ParaVertex*

\begin{proof}
From an instance $\mathcal I=(k \in \mathbb{N},t \in \mathbb{N}, \sigma \in \mathfrak S_k, \sigma_1 \in \mathfrak S_t, \ldots, \sigma_k \in \mathfrak S_t, \mathcal S_A, \mathcal S_B)$, we build a simple polygon $\mathcal P$ with $O(kt+|\mathcal S_A|+|\mathcal S_B|)$ vertices, such that $\mathcal I$ is a YES-instance iff $\mathcal P$ can be guarded by $3k$ vertices.

\textbf{Linker gadget.}
This gadget encodes the 2-elements.
We build a sub-polygon that can be seen entirely by pairs of convex vertices if and only if they correspond to the same 2-element. 

For each $j \in [k]$, permutation $\sigma_j$ will be encoded by a sub-polygon $\mathcal P_j$ that we call \emph{vertex linker}, or simply \emph{linker} (see Figure~\ref{fig:linker-vg}).
We regularly set $t$ consecutive vertices $\alpha^j_1, \alpha^j_2, \ldots, \alpha^j_t$ in this order, along the $x$-axis.
Opposite to this \emph{segment}, we place $t$ vertices $\beta^j_{\sigma_j(1)}, \beta^j_{\sigma_j(2)}, \ldots, \beta^j_{\sigma_j(t)}$ in this order, along the $x$-axis, too.
The $\beta^j_{\sigma_j(1)}, \ldots, \beta^j_{\sigma_j(t)}$, contrary to $\alpha^j_1, \ldots, \alpha^j_t$, are \emph{not} consecutive; we will later add some reflex vertices between them.
At mid-distance between $\alpha^j_1$ and $\beta^j_{\sigma_j(1)}$, to the left, we put a reflex vertex $r^j_{\downarrow}$.
To the left of this reflex vertex, we place a vertical \emph{wall} $d^je^j$ ($r^j_{\downarrow}$, $d^j$, and $e^j$ are three consecutive vertices of $\mathcal P$), so that $\ray(\alpha^j_1,r^j_{\downarrow})$ and $\ray(\alpha^j_t,r^j_{\downarrow})$ both intersect $\seg(d^j,e^j)$. 
That implies that for each $i \in [t]$, $\ray(\alpha^j_i,r^j_{\downarrow})$ intersects $\seg(d^j,e^j)$. 
We denote by $p^j_i$ this intersection. 
The greater $i$, the closer $p^j_i$ is to $d^j$. 
Similarly, at mid-distance between $\alpha^j_t$ and $\beta^j_{\sigma_j(t)}$, to the right, we put a reflex vertex $r^j_{\uparrow}$ and place a vertical wall $x^jy^j$ ($r^j_{\uparrow}$, $x^j$, and $y^j$ are consecutive), so that $\ray(\alpha^j_1,r^j_{\uparrow})$ and $\ray(\alpha^j_t,r^j_{\uparrow})$ both intersect $\seg(x^j,y^j)$. 
For each $i \in [t]$, we denote by $q^j_i$ the intersection between $\ray(\alpha^j_i,r^j_{\uparrow})$ and $\seg(x^j,y^j)$.
The smaller $i$, the closer $q^j_i$ is to $x^j$. 

For each $i \in [t]$, we put around $\beta^j_i$ two reflex vertices, one in $\ray(\beta^j_i,p^j_i)$ and one in $\ray(\beta^j_i,q^j_i)$.
Later we may refer to these reflex vertices as \emph{intermediate reflex vertices}. 
In Figure~\ref{fig:linker-vg}, we merged some reflex vertices but the essential part is that $V(\beta^j_i) \cap \segm(d^j,e^j)=\segm(d^j,p^j_i)$ and $V(\beta^j_i) \cap $ $ \segm(x^j,y^j)=\segm(x^j,q^j_i)$.
Finally, we add a triangular pocket rooted at $g^j$ and supported by $\ray(g^j,\alpha^j_1)$ and $\ray(g^j,\alpha^j_t)$, as well as a triangular pocket rooted at $b^j$ and supported by $\ray(g^j,\beta^j_{\sigma_j(1)})$ and $\ray(g^j,\beta^j_{\sigma_j(t)})$.
This ends the description of the vertex linker (see Figure~\ref{fig:linker-vg}).

\begin{figure}[h]
\centering
\resizebox{350pt}{!}
{
\begin{tikzpicture}[scale=1,extended line/.style={shorten >=-#1,shorten <=-#1}, extended line/.default=1cm]


\foreach \i in {1,...,6}{
\node (t\i) at (\i - 1,-0.4) {$\alpha_\i$} ;
\node[vertex] (v\i) at (\i - 1,0) {} ;
}
\draw[very thick] (v1) -- (v6) ;


\node (tp1) at (0, 6.4) {$\beta_4$} ;
\node (tp2) at (1, 6.4) {$\beta_2$} ;
\node (tp3) at (2, 6.4) {$\beta_5$} ;
\node (tp4) at (3, 6.4) {$\beta_3$} ;
\node (tp5) at (4, 6.4) {$\beta_6$} ;
\node (tp6) at (5, 6.4) {$\beta_1$} ;

\foreach \i in {1,...,6}{
\node[vertex] (vp\i) at (\i-1, 6) {} ;
}


\node[vertex, label=south:$r_{\downarrow}$] (p1) at (-2,3) {} ;
\node[vertex, label=below:$d$] (p2) at (-3,3) {} ;
\node[vertex, label=left:$e$] (p3) at (-3,5) {} ;
\node[vertex, label=north west:$f$] (p4) at (-2,5) {} ;
\node[vertex, label=below:$a$] (u1) at (-0.66,1) {} ;
\node[vertex, label=left:$c$] (u2) at (-1.33,2) {} ;
\draw[very thick] (v1) -- (u1) ;
\draw[very thick] (u2) -- (p1) -- (p2) -- (p3) -- (p4) ;
\path [name path=pocketA] (p2) -- (p3) ;

\node[vertex, label=below:$r_{\uparrow}$] (q1) at (7,3) {} ;
\node[vertex, label=below:$x$] (q2) at (8,3) {} ;
\node[vertex, label=above:$y$] (q3) at (8,5) {} ;
\draw[very thick] (q1) -- (q2) -- (q3) ;
\path [name path=pocketB] (q2) -- (q3) ;

\foreach \i in {1,...,6}{
\path [name path global/.expanded=visa\i] (v\i) -- ($(p1)!-1.8cm!(v\i)$);
\draw [name intersections={of=pocketA and visa\i}, very thin] (v\i) -- (intersection-1) coordinate (e\i) ;
\path [name path global/.expanded=visb\i] (v\i) -- ($(q1)!-1.8cm!(v\i)$);
\draw [name intersections={of=pocketB and visb\i}, very thin] (v\i) -- (intersection-1) coordinate (f\i) ;
}

\draw[very thin] (e1) -- (vp6) -- (f1) ;
\draw[very thin] (e2) -- (vp2) -- (f2) ;
\draw[very thin] (e3) -- (vp4) -- (f3) ;
\draw[very thin] (e4) -- (vp1) -- (f4) ;
\draw[very thin] (e5) -- (vp3) -- (f5) ;
\draw[very thin] (e6) -- (vp5) -- (f6) ;

\path [name path global/.expanded=r1] (e2) -- (vp2) ;
\path [name path global/.expanded=r2] (e5) -- (vp3) ;
\path [name path global/.expanded=r3] (e3) -- (vp4) ;
\path [name path global/.expanded=r4] (e6) -- (vp5) ;
\path [name path global/.expanded=r5] (e1) -- (vp6) ;

\path [name path global/.expanded=s1] (f4) -- (vp1) ;
\path [name path global/.expanded=s2] (f2) -- (vp2) ;
\path [name path global/.expanded=s3] (f5) -- (vp3) ;
\path [name path global/.expanded=s4] (f3) -- (vp4) ;
\path [name path global/.expanded=s5] (f6) -- (vp5) ;

\draw [name intersections={of=r1 and s1}, very thick] (vp1) -- (intersection-1) node[vertex] (ref2) {} ;
\draw [name intersections={of=r2 and s2}, very thick] (vp2) -- (intersection-1) node[vertex] (ref3) {} ;
\draw [name intersections={of=r3 and s3}, very thick] (vp3) -- (intersection-1) node[vertex] (ref4) {} ;
\draw [name intersections={of=r4 and s4}, very thick] (vp4) -- (intersection-1) node[vertex] (ref5) {} ;
\draw [name intersections={of=r5 and s5}, very thick] (vp5) -- (intersection-1) node[vertex] (ref6) {} ;

\foreach \i in {2,...,6}{
\draw[very thick] (ref\i) -- (vp\i) ;
}

\path [name path=r6] (f1) -- (vp6) ;
\draw [name intersections={of=r6 and s5}, very thick] (vp6) -- (intersection-1) node[vertex] (ref7) {} ;
\draw[very thick] (ref7) -- (q3) ;

\path [name path=r0] (p4) -- (vp6) ;
\path [name path=s0] (e4) -- (vp1) ;
\draw [name intersections={of=r0 and s0}, very thick] (vp1) -- (intersection-1) node[vertex,label=above:$h$] (ref1) {} ;

\path [name path=y1] (v1) -- ($(p4)!-2cm!(v1)$);
\path [name path=y2] (v6) -- ($(ref1)!-2cm!(v6)$);
\draw [name intersections={of=y1 and y2}, very thick] (p4) -- (intersection-1) node[vertex,label=left:$g$] (y) {} ;
\draw[very thick] (y) -- (ref1) ;
\draw[dashed, very thin] (v1) -- (p4) ;
\draw[dashed, very thin] (v6) -- (ref1) ;

\path [name path=z1] (vp1) -- ($(u2)!-2.5cm!(vp1)$);
\path [name path=z2] (vp6) -- ($(u1)!-2cm!(vp6)$);
\draw [name intersections={of=z1 and z2}, very thick] (u1) -- (intersection-1) node[vertex,label=left:$b$] (z) {} ;
\draw[very thick] (z) -- (u2) ;
\draw[dashed, very thin] (vp1) -- (u2) ;
\draw[dashed, very thin] (vp6) -- (u1) ;

\foreach \i in {1,...,4}{
\node[inner sep=0.2cm] (ep\i) at (e\i) {} ;
\node (wp\i) at (ep\i.west) {$p_\i$};
}

\foreach \i in {3,...,6}{
\node[inner sep=0.2cm] (fp\i) at (f\i) {} ;
\node (xp\i) at (fp\i.east) {$q_\i$};
}
\end{tikzpicture}
}
\caption{Vertex linker gadget $\mathcal P_j$. We omitted the superscript $j$ in all the labels. Here, $\sigma_j(1)=4,~\sigma_j(2)=2,~\sigma_j(3)=5,~\sigma_j(4)=3,~\sigma_j(5)=6,~\sigma_j(6)=1$.}
\label{fig:linker-vg}
\end{figure}

The following lemma formalizes how exactly the vertices $\alpha^j_i$ and $\beta^j_i$ are linked: say, one chooses to put a guard on a vertex $\alpha^j_i$, then the only way to see $\mathcal P_j$ entirely, by putting a second guard on a vertex of $\{\beta^j_1, \ldots, \beta^j_t\}$ is to place it on the vertex $\beta^j_i$.

\begin{lemma}\label{lem:linker-vg}
For any $j \in [k]$, the sub-polygon $\mathcal P_j$ is seen entirely by $\{\alpha^j_v, \beta^j_w\}$ iff $v = w$.
\end{lemma}

{
\begin{proof}
The regions of $\mathcal P_j$ not seen by $\alpha^j_v$ (i.e., $\mathcal P_j \setminus V(\alpha^j_v)$) consist of the triangles $d^jr^j_{\downarrow}p^j_v$, $x^jr^j_{\uparrow}q^j_v$ and partially the triangle $a^jb^jc^j$.
The triangle $a^jb^jc^j$ is anyway entirely seen by the vertex $\beta^j_i$, for any $i \in [t]$.
It remains to prove that $d^jr^j_{\downarrow}p^j_v \cup x^jr^j_{\uparrow}q^j_v  \subseteq V(\beta^j_w)$ iff $v=w$.

It holds that $d^jr^j_{\downarrow}p^j_v \cup x^jr^j_{\uparrow}q^j_v  \subseteq V(\beta^j_v)$ since, by construction, the two reflex vertices neighboring $\beta^j_v$ are such that $\beta^j_v$ sees $\seg(d^j,p^j_\alpha)$ (hence, the whole triangle $d^jr^j_{\downarrow}p^j_v$) and $\seg(x^j,q^j_\alpha)$ (hence, the whole triangle $x^jr^j_{\uparrow}q^j_v$).   
Now, let us assume that $v \neq w$.
If $v < w$, the interior of the segment $\seg(p_v,p_w)$ is not seen by $\{\alpha^j_v, \beta^j_w\}$, and if $v > w$, the interior of the segment $\seg(q_v,q_w)$ is not seen by $\{\alpha^j_v, \beta^j_w\}$.
\end{proof}
}

The issue we now have is that one could decide to place a guard on a vertex $\alpha^j_i$ and a second guard on a reflex vertex between $\beta^j_{\sigma_j(w)}$ and $\beta^j_{\sigma_j(w+1)}$ (for some $w \in [t-1]$).
This is indeed another way to guard the whole $\mathcal P_j$.
We will now describe a sub-polygon $\mathcal F_j$ (for each $j \in [k]$) called \emph{filter gadget} (see Figure~\ref{fig:filter-vg}) satisfying the property that all its (triangular) pockets can be guarded by adding only one guard on a vertex of $\mathcal F_j$ iff there is already a guard on a vertex $\beta^j_i$ of $\mathcal P_j$.
Therefore, the filter gadget will prevent one from placing a guard on a reflex vertex of $\mathcal P_j$.  
The functioning of the gadget is again based on Lemma~\ref{lem:linking-set-system:pg&vg}.

\textbf{Filter gadget}.
Let $d^j_1, \ldots, d^j_t$ be $t$ consecutive vertices of a regular, say, $20t$-gon, so that the angle made by $\ray(d^j_1,d^j_2)$ and the $y$-axis is a bit below $45^\circ$, while the angle made by $\ray(d^j_{t-1},d^j_t)$ and the $y$-axis is a bit above $45^\circ$. 
The vertices $d^j_1, \ldots, d^j_t$ therefore lie equidistantly on a circular arc~$\mathcal C$.
We now mentally draw two lines $\ell_h$ and $\ell_v$; $\ell_h$ is a horizontal line a bit below $d^j_1$, while $\ell_v$ is a vertical line a bit to the right of $d^j_t$.
We put, for each $i \in [t]$, a vertex $x^j_i$ at the intersection of $\ell_h$ and the tangent to $\mathcal C$ passing through $d^j_i$.
Then, for each $i \in [t-1]$, we set a triangular pocket $\mathcal P(x^j_i)$ rooted at $x^j_i$ and supported by $\ray(x^j_i,d^j_1)$ and $\ray(x^j_i,\beta^j_{\sigma_j(i+1)})$.
For convenience, each point $\beta^j_{\sigma_j(i)}$ is denoted by $c^j_i$ on Figure~\ref{fig:filter-vg}.
We also set a triangular pocket $\mathcal P(x^j_t)$ rooted at $x^j_t$ and supported by $\ray(x^j_t,d^j_1)$ and $\ray(x^j_t,d^j_t)$.
Similarly, we place, for each $i \in [t-1]$, a vertex $y^j_i$ at the intersection of $\ell_v$ and the tangent to $\mathcal C$ passing through $d^j_{i+1}$. 
Finally, we set a triangular pocket $\mathcal P(y^j_i)$ rooted at $y^j_i$ and supported by $\ray(y^j_i,\beta^j_{\sigma_j(i)})$ and $\ray(y^j_i,d^j_t)$, for each $i \in [t-1]$ (see Figure~\ref{fig:filter-vg}).
We denote by $\mathcal P(\mathcal F_j)$ the $2t-1$ triangular pockets of $\mathcal F_j$.

\begin{figure}[h]
\centering
\resizebox{350pt}{!}
{
\begin{tikzpicture}[scale=0.68]
\node[vertex,label=right:$d_1$] (a1) at (0,0) {} ;
\path (a1)--++(65:1) node[vertex,label=right:$d_2$] (a2) {} ;
\path (a2)--++(55:1) node[vertex,label=right:$d_3$] (a3) {} ;
\path (a2)--++(60:0.5) coordinate (aa2) {} ;

\path (a3)--++(45:1) node[vertex,label=below:$d_4$] (a4) {} ;
\path (a3)--++(50:0.5) coordinate (aa3) {} ;

\path (a4)--++(35:1) node[vertex,label=below:$d_5$] (a5) {} ;
\path (a4)--++(40:0.5) coordinate (aa4) {} ;

\path (a5)--++(25:1) node[vertex,label=below:$d_6$] (a6) {} ;
\path (a5)--++(30:0.5) coordinate (aa5) {} ;

\path (a6)--++(20:0.5) coordinate (aa6) {} ;

\draw[very thick] (a1) -- (a2) -- (a3) -- (a4) -- (a5) -- (a6) ;

\node[vertex,label=below:$x_1$] (x1) at (0,-1) {} ;
\draw[very thick] (a1) -- (x1) ;
\path [name path = c] (0,-1) -- (-9,-1) ;
\path [name path= d1] (a2) -- ($(aa2)!3cm!(a2)$);
\path [name path= d2] (a3) -- ($(aa3)!5cm!(a3)$);
\path [name path= d3] (a4) -- ($(aa4)!7cm!(a4)$);
\path [name path= d4] (a5) -- ($(aa5)!10cm!(a5)$);
\path [name path= d5] (a6) -- ($(aa6)!15cm!(a6)$);

\foreach \i [count=\j from 2] in {1,...,5}{
\draw [name intersections={of=c and d\i}, dotted] (a\j) -- (intersection-1) node[vertex,label=below:$x_\j$] (x\j) {} ;
}

\begin{scope}[opacity=0.3,yshift=1cm, xshift=-2cm, rotate=0]
\foreach \i in {1,...,6}{
\node[vertex,label=above:$c_\i$] (b\i) at (-7+\i,9) {};
}
\node[vertex] (zr1) at (-5.4,8.94) {} ;
\node[vertex] (zr2) at (-4.5,8.94) {} ;
\node[vertex] (zr3) at (-3.6,8.94) {} ;
\node[vertex] (zr4) at (-2.4,8.94) {} ;
\node[vertex] (zr5) at (-1.5,8.94) {} ;

\foreach \i [count=\j from 2] in {1,...,5}{
\draw[very thick] (b\i) -- (zr\i) ;
\draw[very thick] (zr\i) -- (b\j) ;
}
\end{scope}

\path [name path = g] (0,-0.8) -- (-9,-0.8) ;

\foreach \i [count=\j from 2] in {1,...,5}{
\draw[name path global/.expanded=f\i,very thin] (x\i) -- (b\j) ;
}
\path [name path = f6] (x6) -- (a6) ;

\foreach \i in {2,...,6}{
\draw [name path global/.expanded= h\i, very thin] (x\i) -- (a1) ;
}

\foreach \i in {2,...,6}{
\draw [name intersections={of=g and h\i}, very thick] (x\i) -- (intersection-1) coordinate (z\i) {} ;
}

\foreach \i in {1,...,6}{
\draw [name intersections={of=g and f\i}, very thick] (x\i) -- (intersection-1) coordinate (w\i) {} ;
}

\foreach \i [count=\j from 2] in {1,...,5}{
\draw[very thick] (w\i) -- (z\j) ;
}

\path [name path = l] (4.5,3.4) -- (4.5,9) ;
\path [name path= m2] (a2) -- ($(aa2)!-8cm!(a2)$);
\path [name path= m3] (a3) -- ($(aa3)!-6cm!(a3)$);
\path [name path= m4] (a4) -- ($(aa4)!-4cm!(a4)$);
\path [name path= m5] (a5) -- ($(aa5)!-2cm!(a5)$);
\draw [name intersections={of=l and m2}, dotted] (a2) -- (intersection-1) node[vertex,label=right:$y_1$] (y1) {} ;
\draw [name intersections={of=l and m3}, dotted] (a3) -- (intersection-1) node[vertex,label=right:$y_2$] (y2) {} ;
\draw [name intersections={of=l and m4}, dotted] (a4) -- (intersection-1) node[vertex,label=right:$y_3$] (y3) {} ;
\draw [name path= k5,name intersections={of=l and m5}, dotted] (a5) -- (intersection-1) node[vertex,label=right:$y_4$] (y4) {} ;

\node[vertex,label=right:$y_5$] (y5) at (4.5,3.4) {};
\draw[very thick] (y5) -- (a6) ;

\path [name path = p] (4.3,3.4) -- (4.3,10) ;
\foreach \i [count=\j from 2] in {1,...,5}{
\draw [name path global/.expanded= j\j, very thin] (y\i) -- (b\i);
\draw [name intersections={of=p and j\j}, very thick] (y\i) -- (intersection-1) coordinate (za\j) {} ;
}

\foreach \i [count=\j from 2] in {1,...,4}{
\draw[name path global/.expanded= k\j,very thin] (y\i) -- (a6) ;
\draw [name intersections={of=p and k\j}, very thick] (y\i) -- (intersection-1) coordinate (zb\j) {} ;
}

\foreach \i [count=\j from 3] in {2,...,5}{
\draw[very thick] (za\j) -- (zb\i) ;
}
\end{tikzpicture}
}
\caption{The filter gadget $\mathcal F_j$. Again, we omit the superscript $j$ on the labels. Vertices $c_1, c_2, \ldots, c_t$ are not part of $\mathcal F_j$ and are in fact the vertices $\beta^j_{\sigma_j(1)}, \beta^j_{\sigma_j(2)}, \ldots, \beta^j_{\sigma_j(t)}$ and the vertices in between the $c_i$'s are the reflex vertices that we have to \emph{filter out}.}
\label{fig:filter-vg}
\end{figure}

\begin{lemma}\label{lem:filter-vg}
For each $j \in [k]$, the only ways to see $\mathcal P(\mathcal F_j)$ and the triangle $a^jb^jc^j$ entirely with only two guards on vertices of $\mathcal P_j \cup \mathcal P(\mathcal F_j)$ is to place them on vertices $c^j_i$ and $d^j_i$ $($for any $i \in [t]$$)$.
\end{lemma}

{
\begin{proof}
Proving this lemma will, in particular, entail that it is not possible to see $\mathcal P(\mathcal F_j)$ entirely with only two vertices if one of them is a reflex vertex between $c^j_i$ and $c^j_{i+1}$.
We recall that such a vertex is called an intermediate reflex vertex (in color class $j$).
Because of the pocket $a^jb^jc^j$, one should put a guard on a $c^j_i$ (for some $i \in [t]$) or on an intermediate reflex vertex in class $j$.
As vertices $a^j$, $b^j$, and $c^j$ do not see anything of $\mathcal P(\mathcal F_j)$, placing the first guard at one of those three vertices cannot work as a consequence of what follows.

Say, the first guard is placed at $c^j_i$ ($=\beta^j_{\sigma(i)}$). 
The pockets $\mathcal P(x^j_1), \mathcal P(x^j_2), \ldots, \mathcal P(x^j_{i-1})$ and $\mathcal P(y^j_i),$ $ \mathcal P(y^j_{i+1}), \ldots, \mathcal P(x^j_{t-1})$ are entirely seen, while the vertices $x^j_i, x^j_{i+1}, \ldots, x^j_t$ and $y^j_1, y^j_2, \ldots,$ $y^j_{i-1}$ are not.
The only vertex that sees simultaneously all those vertices is $d^j_i$.
The vertex $d^j_i$ even sees the \emph{whole} pockets $\mathcal P(x^j_i), \mathcal P(x^j_{i+1}), \ldots, \mathcal P(x^j_t)$ and $\mathcal P(y^j_1), \mathcal P(y^j_2), \ldots,$ $\mathcal P(y^j_{i-1})$.
Therefore, all the pockets $\mathcal P(\mathcal F_j)$ are fully seen.

Now, say, the first guard is put on an intermediate reflex vertex $r$ between $c^j_i$ and $c^j_{i+1}$ (for some $i \in [t-1]$).
Both vertices $x^j_i$ and $y^j_i$, as well as $x^j_t$, are not seen by $r$ and should therefore be seen by the second guard.
However, no vertex simultaneously sees those three vertices.  
\end{proof}
}

\textbf{Putting the pieces together.}
The permutation $\sigma$ is encoded the following way. 
We position the vertex linkers $\mathcal P_1, \mathcal P_2,$ $\ldots, \mathcal P_k$ such that $\mathcal P_{i+1}$ is below and slightly to the left of $\mathcal P_i$. 
Far below and to the right of the $\mathcal P_i$'s, we place the $\mathcal F_i$'s such that the uppermost vertex of $\mathcal F_{\sigma(i)}$ is close and connected to the leftmost vertex of $\mathcal F_{\sigma(i+1)}$, for all $i \in [t-1]$.
We add a constant number of vertices in the vicinity of each $\mathcal P_j$, so that the only filter gadget that vertices $\beta^j_1, \ldots, \beta^j_t$ can see is $\mathcal F_j$ (see Figure~\ref{fig:overall-vg}).
Similarly to the point guard version, we place vertically and far from the $\alpha^j_i$'s, one triangular pocket $\mathcal P(z_{A,q})$ rooted at vertex $z_{A,q}$ and supported by $\ray(z_{A,q},\alpha^j_i)$ and $\ray(z_{A,q},\alpha^{j'}_{i'})$, for each $A$-interval $I_q=[a^j_i,a^{j'}_{i'}] \in \mathcal S_A$ (Track $1$).
Finally, we place vertically and far from the $d^j_i$'s, one triangular pocket $\mathcal P(z_{B,q})$ rooted at vertex $z_{B,q}$ and supported by $\ray(z_{B,q},d^j_i)$ and $\ray(z_{B,q},d^{j'}_{i'})$, for each $B$-interval $I_q=[b^j_{\sigma_j(i)},b^{j'}_{\sigma_{j'}(i')}] \in \mathcal S_B$ (Track $2$).
We make sure that, all projected on the $x$-axis, $\mathcal F_{\sigma(1)}$ is to the right of $\mathcal P_1$ and to the left of Track 1, so that, for every $i \in [t]$, the vertex $d_i^{\sigma(1)}$ sees the top edge of the gallery entirely.
This ends the construction (see Figure~\ref{fig:overall-vg}).

\begin{figure}[h]
\centering
\resizebox{300pt}{!}
{
\includegraphics{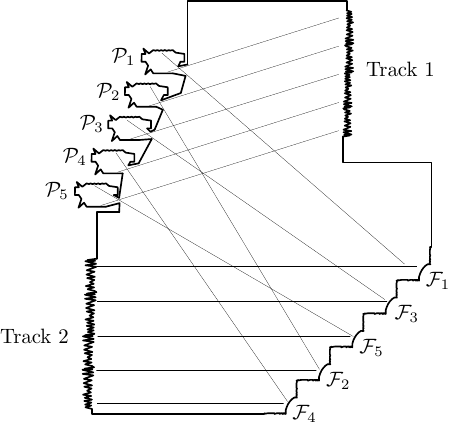}
}
\caption{Overall picture of the reduction with $k=5$, and $\sigma=42531$.
  The linker gadgets $\mathcal P_j$, together with $\mathcal F_j$, force guards at vertices $\alpha_{i(j)}^j, \beta_{i(j)}^j$.
  The filter gadgets $\mathcal F_j$ transmit the choice of $\beta_{i(j)}^j$ and ensure that no other guard placement can be made in $\mathcal P_j$.
  The $A$-intervals of the \shs instance are encoded by triangular pockets on Track 1, while the $B$-intervals are encoded on Track 2.}
\label{fig:overall-vg}
\end{figure}

\textbf{Correctness.}
We now prove the correctness of the reduction.
Assume that $\mathcal I$ is a YES-instance and let $\{(a^1_{s_1},b^1_{s_1}), \ldots, (a^k_{s_k},b^k_{s_k})\}$ be a solution.
We claim that the set of vertices $G=\{\alpha^1_{s_1}, \beta^1_{s_1}, d^1_{\sigma^{-1}_1(s_1)},$ $ \ldots, $ $\alpha^k_{s_k}, \beta^k_{s_k}, d^k_{\sigma^{-1}_k(s_k)}\}$ guards the whole polygon $\mathcal P$.
Let $z^j := d^j_{\sigma^{-1}_j(s_j)}$ for notational convenience.
By Lemma~\ref{lem:linker-vg}, for each $j \in [k]$, the sub-polygon $\mathcal P_j$ is entirely seen, since there are guards on $\alpha^j_{s_j}$ and $\beta^j_{s_j}$.
By Lemma~\ref{lem:filter-vg}, for each $j \in [k]$, all the pockets of $\mathcal F_j$ are entirely seen, since there are guards on $\beta^j_{s_j}=c^j_{\sigma^{-1}_j(s_j)}$ and $d^j_{\sigma^{-1}_j(s_j)}=z^j$.
For each $A$-interval (resp.~$B$-interval) in $\mathcal S_A$ (resp.~$\mathcal S_B$) there is at least one $2$-element $(a^j_{s_j},b^j_{s_j})$ such that $a^j_{s_j} \in \mathcal S_A$ (resp. $b^j_{s_j} \in \mathcal S_B$).
Thus, the corresponding pocket is guarded by $\alpha^j_{s_j}$ (resp.~$\beta^j_{s_j}$).
The rest of the polygon is seen by, for instance, $z^{\sigma(1)}$ and $z^{\sigma(k)}$.

We now assume that there is a set $G$ of $3k$ vertices guarding $\mathcal P$.
We will show that $\mathcal I$ is a YES-instance.
For each $j \in [k]$, vertices $b^j$, $g^j$, and $x^j_t$ are seen by three pairwise-disjoint sets of vertices.
The first two sets are contained in the vertices of sub-polygon $\mathcal P_j$ and the third one is contained in the vertices of $\mathcal F_j$.
Therefore, to see $\mathcal P_j \cup \mathcal P(\mathcal F_j)$ entirely, three vertices are necessary.
Summing that over the $k$ color classes, this corresponds already to $3k$ vertices which is the size of~$G$.
Thus, $G$ contains a set $S_j$ of \emph{exactly} $3$ guards among the vertices of $\mathcal P_j \cup \mathcal P(\mathcal F_j)$.

The guard of $S_j$ responsible for seeing $g^j$ does not see $b^j$ nor any pockets of $P(\mathcal F_j)$.
Hence there are only two guards of $S_j$ performing the latter task. 
Therefore, by Lemma~\ref{lem:filter-vg}, there should be an $s_j \in [t]$ such that both $d^j_{s_j}$ and $c^j_{s_j}=\beta^j_{\sigma_j(s_j)}$ are in $G$.
The only vertices seeing $g^j$ are $f^j, g^j, h^j$ and $a^j_1, \ldots, a^j_t$.
As $d^j_{s_j}$ and the $3k-3$ guards of $G \setminus S_j$ do not see the edges $d^je^j$ and $x^jy^j$ at all, by Lemma~\ref{lem:linker-vg}, among $a^j_1, \ldots, a^j_t$ the only possibility for the third guard of $S_j$ is $\alpha^j_{\sigma_j(s_j)}$.
We can assume that the third guard of $S_j$ is indeed $\alpha^j_{\sigma_j(s_j)}$, since $f^j, g^j, h^j$ do not see any pockets outside of $\mathcal P_j$ (whereas $\alpha^j_{\sigma_j(s_j)}$, in principle, does in Track 1).

So far, we showed that $G$ is of the form $\{\alpha^1_{\sigma_1(s_1)}, \beta^1_{\sigma_1(s_1)}, d^1_{s_1}, \ldots, \alpha^j_{\sigma_j(s_j)}, \beta^j_{\sigma_j(s_j)}, d^j_{s_j}, \ldots, \alpha^k_{\sigma_k(s_k)},$ $ \beta^k_{\sigma_k(s_k)}, d^k_{s_k}\}$.
It means that $\alpha^1_{\sigma_1(s_1)}, \ldots, \alpha^k_{\sigma_k(s_k)}$ see all the pockets of Track 1, while $d^1_{s_1}, \ldots, d^k_{s_k}$ see all the pockets of Track 2.
Therefore the set of $k$ 2-elements $\{(a^1_{\sigma_1(s_1)},b^1_{\sigma_1(s_1)}), \ldots, (a^k_{\sigma_k(s_k)},b^k_{\sigma_k(s_k)})\}$ is a hitting set of both $\mathcal S_A$ and $\mathcal S_B$, hence $I$ is a YES-instance.

Let us bound the number of vertices of $\mathcal P$.
Each sub-polygon $\mathcal P_j$ or $\mathcal F_j$ contains $O(t)$ vertices. 
\emph{Track $1$} contains $3|\mathcal S_A|$ vertices and \emph{Track $2$} contains $3|\mathcal S_B|$ vertices. 
Linking everything together requires $O(k)$ additional vertices.
So, in total, there are $O(kt+|\mathcal S_A|+|\mathcal S_B|)$ vertices.
Thus, this reduction together with Theorem~\ref{cor:main} implies that \vgag is $\wone$-hard and cannot be solved in time $f(k)n^{o(k / \log k)}$, where $n$ is the number of vertices of the polygon and $k$ the number of guards, for any computable function $f$, unless the ETH fails.
\end{proof}



\end{document}